\documentclass[11pt,papersize={8.5in,11in}]{article}

\usepackage{amssymb,amsmath,amsthm}
\usepackage[left=2.54cm, right=2.54cm]{geometry}

\usepackage[algosection,algoruled,vlined,nofillcomment,linesnumbered]{algorithm2e}
\SetEndCharOfAlgoLine{}
\usepackage[shortlabels]{enumitem}
\usepackage[colorlinks=true,citecolor=blue,urlcolor=blue,linkcolor=blue,bookmarksopen=true,unicode=true,pdffitwindow=true]{hyperref}
\usepackage{todonotes}
\renewcommand{\O}[1]{\ensuremath{\mathcal{O}{\left( #1 \right)}}}

{{\bfseries #1} {\itshape (#2)}.}{}

\newcommand{\set}[1]{\ensuremath{\left\{#1 \right\}}}

\newtheorem{theorem}{Theorem}[section]

\newtheorem{lemma}[theorem]{Lemma}
\newtheorem{proposition}[theorem]{Proposition}
\newtheorem{definition}[theorem]{Definition}
\newtheorem{remark}[theorem]{Remark}
\newtheorem{example}[theorem]{Example}

\SetKw{KwGoTo}{go to}
\SetKwFunction{Enumerate}{Enumerate}
\SetKwFunction{Visit}{Visit}
\SetKwFunction{Ready}{MakeReady}
\SetKwFunction{Trim}{Trim}
\SetKwFunction{Main}{Main}
\SetKwProg{myproc}{Procedure}{}{}

\newcommand{\dcup}{\operatorname{\dot{\cup}}}
\newcommand{\alg}{\mathcal{A}}
\newcommand{\halg}{\hat{\alg}}
\newcommand{\flip}{\mathcal{F}}
\newcommand{\id}{\mathcal{I}}

\newcommand{\dittotikz}{%
    \tikz{
        \draw [line width=0.12ex] (-0.2ex,0) -- +(0,0.8ex)
            (0.2ex,0) -- +(0,0.8ex);
        \draw [line width=0.08ex] (-0.6ex,0.4ex) -- +(-1.5em,0)
            (0.6ex,0.4ex) -- +(1.5em,0);
    }%
}

\title{Constant delay Gray code enumeration of ideals and antichains in posets}
\author{Sofia Brenner\footnote{Department of Mathematics and Natural Sciences, University of Kassel, Germany.\\ Present address: Department of Applied Mathematics, Charles University, Prague, Czech Republic.\\ E-mail address: \texttt{sbrenner@mathematik.uni-kassel.de}} \ and Ji\v{r}í Fink\footnote{Department of Theoretical Computer Science and Mathematical Logic, Charles University, Prague, Czech Republic. E-mail address: \texttt{fink@ktiml.mff.cuni.cz}}}
\date{}

\begin{document}
	\maketitle
	\begin{abstract}
		We present an algorithm that enumerates all ideals of an input poset with constant delay in Gray code order, i.e., such that consecutively visited ideals differ in at most three elements.
		This answers a long-standing open problem posed by Pruesse and Ruskey, and improves upon previous algorithms by Pruesse and Ruskey, Squire, Habib, Medina, Nourine and Steiner, as well as Abdo. 
		Using the same techniques, we also obtain an algorithm that enumerates all antichains of an input poset with constant delay such that successively visited antichains differ in at most three elements. 
		
	As a key technical ingredient, we introduce a new potential-based analysis framework for recursive algorithms, which we call the \emph{Pyramid method}. 
	We show that this method subsumes the Push-out method of Uno. 
	Beyond the present application, the Pyramid method is a general framework to analyze recursive algorithms and may thus be of independent interest.
	\end{abstract}
	
	\section{Introduction}
		
	\subsection{Combinatorial enumeration}
	
	The aim of an \emph{enumeration algorithm} is to list all objects in a (finite) set $X$ exactly once.
	The set $X$ is usually given by an implicit description and the task of the algorithm is to compute its elements.  
	Clearly, the goal are algorithms that enumerate $X$ efficiently while requiring small amounts of memory.
	Our focus lies on enumerating combinatorial objects, such as substructures of graphs, posets, or subsets of permutations. 
	
	The enumeration of combinatorial structures has a long history. 
	The most basic algorithm to enumerate a set $X$ identifies~$X$ with the set of all strings of length $n$ over an alphabet $\Sigma$. It non-deterministically chooses a letter of $\Sigma$ for all $n$ positions and accepts the resulting string if it belongs to $X$.
	An early paper by Floyd \cite{floyd1967nondeterministic} converts this idea into a deterministic backtracking algorithm, using pruning techniques and a strategy for choosing the right string to improve the performance. 
	Read and Tarjan~\cite{read1975bounds} enumerated graph substructures, such as spanning trees, cycles, and various kinds of paths. 
	Their method was later referred to as the \emph{binary partition method} or \emph{Flashlight search} (for instance, see~\cite{merino2024traversing}).
		Fukuda and Matsui \cite{fukuda1994finding} enumerated all perfect matchings in bipartite graphs in time $\O{m|X|}$ and space $\O{nm}$.
		Avis and Fukuda~\cite{avis1996reverse} introduced the concept of \emph{reversed search}. 
		They apply it to enumerate all triangulations of a set of~$n$ points in the plane, spanning trees of a graph, and topological orderings of an acyclic graph.
		Recent work includes the \emph{proximity search} framework by Conte, Grossi, Marino, Uno, and Versari~\cite{CON19,DBLP:journals/siamcomp/ConteGMUV22}, which they use to enumerate, for instance, maximal subgraphs or maximal trees in a given graph.
		Merino and M\"utze~\cite{merino2024traversing} introduced a framework to list all vertices of $0/1$-polytopes via combinatorial optimization. 
		This yields enumeration algorithms for numerous combinatorial objects, including bases and independent sets in a matroid, spanning trees, forests,
		matchings, or vertex covers.
		Fink~\cite{fink2025constant} derived a constant delay algorithm for listing perfect matchings.
	
	For many enumeration algorithms, especially in combinatorial enumeration, the size of the enumerated set $X$ is exponential in the input size.
One thus rather measures the computation time spent between two successively visited elements. 
		An \emph{instruction} is an elementary operation executed in a constant time in a computation model, such as a deterministic random access machine or pointer machine.
	Let $t(i)$ be the number of instructions executed before the $i$-th element is visited.
	The \emph{(worst-case) delay} of an enumeration algorithm is the maximal time spent between two successively visited elements, i.e., $\max_{i \in \{1, \dots, |X|-1\}}\{t(i+1) - t(i)\}$. 
	Clearly, the optimal outcome are algorithms with constant delay. 
	These algorithms are often called \emph{loopless}.
	As in many applications, the time between two outputs differs drastically, one often considers an average. 
	The \emph{amortized delay} of an enumeration algorithm is the average computation time spent between two successive outputs, i.e., the total computation time of the algorithm divided by the number of its outputs. 
	Note that some authors call this the \emph{average delay}, and define the amortized delay of an algorithm as $\max_{i \in \{1, \dots, |X|\}} t(i)/i$. 
    In this paper, we will also state bounds for this cardinality.

	In many situations, one requires that successive outputs of an enumeration algorithm differ only by a small change, a so-called \emph{flip}.
	The definition of a flip depends on the problem.  
	A listing of the enumerated elements respecting this condition is called a \emph{combinatorial Gray code}. 
	A combinatorial generation problem on a set $X$ can be described via its \emph{flip graph}. 
	The vertices of the flip graph are the elements of $X$, and two elements are connected whenever they differ by a flip. 
	Finding a combinatorial Gray code thus amounts to finding a Hamiltonian path in the flip graph.
	The classical \emph{binary reflected Gray code}~\cite{gray_1953} lists all binary strings of a fixed length~$n$ such that each two consecutive strings differ in a single bit. 
	The \emph{Steinhaus-Johnson-Trotter algorithm}~\cite{DBLP:journals/cacm/Trotter62,MR0159764,MR0157881} lists permutations of length~$n$ such that two successively listed permutations differ in a transposition of adjacent values. 
	In recent years, the area of combinatorial Gray codes has seen a flurry of results. 
	A main line of research is the \emph{permutation language framework} introduced by M\"utze and collaborators (see~\cite{HAR20}).
	They proved the existence of combinatorial Gray codes for numerous classes of combinatorial objects, such as certain classes of pattern-avoiding permutations~\cite{HAR20}, quotients of the weak order~\cite{HOA21}, rectangulations~\cite{MER23}, elimination trees~\cite{cardinal2024combinatorial}, acyclic orientations~\cite{DBLP:conf/soda/CardinalHMM23}, binary trees~\cite{DBLP:journals/ejc/GregorMN24}, or graphs of regions of hyperplane arrangements~\cite{brenner2026}. 
	In many cases, the construction of the Gray codes is explicit and can be turned into an algorithm, but not in all cases a loopless algorithm is known.
	We refer to M\"utze's survey~\cite{mutze2022combinatorial} for a comprehensive overview on combinatorial Gray codes.

	\subsection{Enumeration algorithms for posets}

	A \emph{partially ordered set (poset)} is a pair $(P, \preceq)$ consisting of a set $P$ and a \emph{partial order} $\preceq$, that is, a reflexive, antisymmetric, transitive relation, on $P$.
	In this paper, we only consider finite posets. 
	If the relation $\preceq$ is understood from the context, we simply write $P$ for the poset $(P, \preceq)$.
	For $u,v \in P$ with $u \preceq v$ and $u \neq v$, we write $u \prec v$ and say that $u$ is \emph{smaller} than~$v$, or that $v$ is \emph{larger} than~$u$. 
	The elements $u$ and $v$ are \emph{incomparable} if neither $u \preceq v$ nor $v \preceq u$ holds. 
	We then call the set $\{u,v\}$ an \emph{incomparable pair}.
	We say that $v$ \emph{covers} $u$, or that $u$ and~$v$ are in a \emph{cover relation}, if $u \prec v$ and there is no $z \in P$ with $u \prec z \prec v$. 
	A useful tool to represent a poset is its \emph{Hasse diagram}. 
	This is a graph drawn in the plane whose vertices are the elements of~$P$. 
	If $v$ covers $u$, we connect $u$ and $v$ by an edge, and draw $v$ higher than $u$. 
	See Figure~\ref{fig:piexample}\,(a) for an illustration.
	The \emph{upset} and \emph{downset} of $x \in P$ consist of all elements $u \in P$ with $x \preceq u$ or $u \preceq x$, respectively.
	A subset $I$ of $P$ is an \emph{ideal} if if it is \emph{downward-closed}, i.e., for all $u \in I$, the downset of $u$ is contained in $I$.
	The set of ideals of $P$ will be denoted by $\id(P)$.
	A \emph{chain} in~$P$ is a subset consisting of pairwise comparable elements of $P$,
	an \emph{antichain} is a subset consisting of pairwise incomparable elements.
	An antichain with precisely~$k$ elements is called a \emph{$k$-antichain}.
	
Posets play a prominent role both in combinatorics and algorithms, and hence their algorithmic aspects have been widely studied. 
For instance, it was shown that various counting problems in posets are \#P-complete, such as counting the number of antichains~\cite{provan_complexity_1983} (and thus ideals, since every poset contains the same number of ideals and antichains), or linear extensions of a given poset~\cite{brightwell1991counting}. 
Pruesse and Ruskey~\cite{pruesse1994generating} derived an algorithm to enumerate linear extensions of an input poset such that successively enumerated extensions differ by one or two adjacent transpositions. 
This was turned into a loopless algorithm by Canfield and Williamson~\cite{CAN95}. 

In contrast, the loopless enumeration of ideals has been a long-standing open problem. 
We refer to the discussions of this question in \cite{pruesse1993gray,pruesse1994generating, habib2001,mutze2022combinatorial,powers2025} for an overview. 
The fastest enumeration algorithms for ideals in posets, either listing in Gray code order or in arbitrary order, have logarithmic delay.
For $k \in \mathbb{N}$, we say that a listing is a \emph{$k$-Gray code} if two consecutively visited ideals differ in at most $k$ elements (i.e, their symmetric difference has size at most $k$).
In the literature, the usual definition of a Gray code order is a 2-Gray code, i.e, successively listed ideals differ in at most two elements. 
In other words, one ideal arises from the other by removing, adding, or exchanging a single element. 
In our algorithm, we derive a 3-Gray code, i.e, successively visited ideals may differ in at most three elements.
Note that this does not change the asymptotic time complexity of the algorithms.

Early algorithms~\cite{lawler1979,schrage1987,provan_complexity_1983} list the ideals of a poset of size $n$ with delay $\O{n^2}$. 
Probably the first algorithm with average delay $\O{n}$ was given by Steiner~\cite{steiner1986}.
Pruesse and Ruskey~\cite{pruesse1993gray} obtain an algorithm for a (2-)Gray code listing that runs in time $\O{|\id(P)| n}$ on general posets, using a more general framework for antimatroids.
However, their algorithm lists every element twice.
Medina and Nourine~\cite{medina1994} gave an algorithm with amortized delay $\O{d^+(P)}$, where $d^+(P)$ denotes the maximal out-degree in the Hasse diagram of $P$.
Squire~\cite{squire1995} (see also~\cite{powers2025}) gives an algorithm with average delay $\O{\log n}$, however not listing the elements in Gray code order. 
Habib, Medina, Nourine and Steiner~\cite{habib2001} present a different algorithm as part of their framework to enumerate distributive lattices, which lists the ideals with amortized delay $\O{n}$ in Gray code order.
Further algorithms to enumerate ideals in Gray code order with amortized delay $\O{n}$ were given by Abdo~\cite{abdo2009,abdo2010,abdo2013}.
Merino and M\"utze~\cite{DBLP:conf/focs/MerinoM23,merino2024traversing} mention that their much more general framework to enumerate the vertices of $0/1$-polytopes yields an enumeration algorithm for poset ideals as a special case. 
However, this algorithm has delay $\O{n^4 \log n}$.
For antichains, their framework also yields an algorithm, albeit with delay $\O{n^{2.5} \log n}$.
Constant average delay algorithms are known for special poset classes, such as forest posets~\cite{koda1993} or interval posets~\cite{habib1997}.

Many of the above-mentioned algorithms, for instance those by Squire and Abdo, rely on a recursive strategy that splits the ideals of the input poset $P$ into two subsets. 
One consists of the ideals which contain a fixed element $x \in P$ (and thus also its downset), and the other of the ideals not containing $x$ (and thus none of the elements of its upset). 
The main difference lies in the choice of the element $x$. 
If the input poset is a chain, it is easy to see that this strategy yields at least logarithmic delay. 
In order to obtain a constant delay algorithm, a different recursive strategy is therefore needed.

We conclude this part by mentioning a variant of the above problem, the generation of ideals of a fixed size in a poset. 
Here, sucessively visited ideals differ in a swap of an element. 
This variant has attracted considerable attention over the years. 
In~\cite{koda1993}, it was remarked without proof that their algorithm generating all ideals in forest posets can also be used to generate ideals of a fixed size. 
In~\cite{pruesse1993gray}, this was then studied explicitly. 
Wild~\cite{wild2014} gave an algorithm with $\O{n^3}$ delay.
This was recently improved by Coumes, Bouadi, Nourine, and Termier~\cite{coumes2021}, who used this to enumerate so-called \emph{skyline groups} used in decision-making. 
They showed that these correspond to ideals of fixed size in certain posets. 
In this paper, they gave an algorithm that visits all ideals of size $k$ in a poset with delay $\O{\omega^2}$, where $\omega$ denotes the \emph{width} of the input poset.
For further applications of the enumeration of ideals of fixed size, such as for codes and covering arrays, we refer to Powers' PhD thesis~\cite{powers2025}.

	\subsection{Our results}

	The aim of this paper is to provide an algorithm for the enumeration of poset ideals as well as antichains with constant amortized delay.
	This paper has two parts: in the first part, we develop a method for potential analysis, the \emph{Pyramid condition}, to analyze the amortized time complexity of enumeration algorithms.
	In the second part, we describe an enumeration algorithm for ideals in a poset in $3$-Gray code order. 
	We use the Pyramid condition to show that it has constant delay. 
	Using the same technique, we also obtain a constant delay algorithm for enumerating antichains. 
		
	\subsubsection{Complexity analysis and the Pyramid condition}
	
	We first develop a new variant of the potential analysis, which we call the \emph{Pyramid method}. 
	Our main interest is to analyze combinatorial enumeration algorithms. 
	However, the concept is not specifically tailored to this situation. 
	The method exploits the idea that in a recursion tree of an enumeration algorithm, the leaves correspond to iterations generating output in quick succession while requiring only small computation time whereas the opposite occurs for the inner nodes.
	The idea of the amortization is thus to use the excess computation time of the leaves for the computations in nodes higher in the tree. 
	To this end, we use a potential measuring the computation time that a child node can spare for its parent.  
	An illustration for a simple algorithm enumerating perfect matchings in an input graph is given in Figure~\ref{fig:recursiontree}.
		
	For an iteration $X$, let $T(X)$ denote the number of instructions needed for the execution of iteration $X$, $C(X)$ denotes the set of iterations recursively called by $X$ (i.e., the \emph{children} of $X$), and $X'$ denotes the set of elements directly visited by $X$.
	We refer to Section~\ref{sec:potentialanalysis} for full definitions. 
	The above idea is captured by the following definition:
	
	\begin{definition}[Pyramid condition]\label{def:pyramidconditioninto}
			We say that a recursive algorithm $\alg$ satisfies the \emph{Pyramid condition} with respect to a potential function $\Phi$ and time $T^\star$ if there exists a constant $\mu > 0$ such that for every iteration $X$ of~$\alg$, it holds that \[\sum_{Y \in C(X)} \Phi(Y) + \mu |X'| - \Phi(X) \ge \frac{T(X)}{T^\star}.\]
	\end{definition}

Let $\Delta$ be an upper bound for the number of instructions on a path from the root to a leaf in the recursion tree. 
Our main result in this part is the following:
	
	\begin{theorem} \label{thm:pyramid_intro}
	Let $\mathcal{A}$ be a recursive algorithm enumerating a set $X$ and assume that $\mathcal{A}$ satisfies the Pyramid condition for a potential function $\Phi$ and time $T^\star$. 
	The total computation time to enumerate $X$ is $\O{|X| T^\star - \Phi(X)}$, so the amortized delay of $\alg$ is $\O{T^\star}$.
	The first~$k$ elements of~$X$ are visited in time $\O{k T^\star + \Delta}$. 
	\end{theorem}

In the time complexity to enumerate the first $k$ elements in Theorem~\ref{thm:pyramid_intro}, the additive constant~$\Delta$ can be viewed as a preprocessing time needed at the beginning, whereas each visited element requires~$T^\star$ additional time.

We now turn from amortized delay to worst-case delay analysis.
It is well-known that an algorithm with amortized delay~$T^\star$ can be turned into an algorithm with worst-case delay~$T^\star$, but at the possible price of exponential memory consumption.
Using an idea that was originally presented by Uno in an unpublished technical report~\cite{uno2003}, we obtain the following de-amortization result for algorithms satisfying the Pyramid condition:

\begin{theorem}[Worst-case delay] \label{thm:loopless_intro}
Let $\alg$ be a recursive algorithm enumerating a set $X$ that satisfies the Pyramid condition with respect to a potential $\Phi$ and time $T^\star$. 
Assume that the change between two successively visited elements of $\alg$ can be stored using $\O{1}$ memory and that these changes can be applied in time $\O{T^\star}$. 
Then there exists an Algorithm $\halg$ which enumerates $X$ in the same order as $\alg$ with worst-case delay $\O{T^\star}$, while requiring $\O{\Delta/T^\star}$ additional space and $\O{\Delta}$ additional time for the initialization.
\end{theorem}

Whereas there is no general upper bound for $\Delta$, in many applications, the depth of the recursion tree as well as the number of instructions along a longest path from a root to the leaf is polynomially bounded in the input size.
In this case, both the additional space requirement as well as the initialization time for the algorithm $\alg'$ from Theorem~\ref{thm:loopless_intro} are polynomial in the input size.
As a typical example, consider an easy enumeration algorithm, for instance for spanning trees, that in each iteration chooses an object of the input graph, such as an edge to be included or excluded. 
Clearly, both the depth of this recursion tree and the number of iterations on a path from the root to a leaf are at most linear in the input size.

\paragraph{Comparison with the Push-out method} 
Uno~\cite{uno2015constant} devised a method to analyze the amortized time complexity of recursive algorithms, the so-called \emph{Push-out method}. 
He employs it to show that various combinatorial enumeration algorithms, for instance, for enumerating elimination orderings, matchings forests, connected induced subgraphs, or spanning trees in a given graph, have constant amortized delay. 
We prove that the Pyramid method subsumes the Push-out method: 

\begin{theorem}\label{thm:pushout_intro}
	If a recursive enumeration algorithm $\alg$ satisfies Uno's Push-out condition with time $T^\star$ (i.e., the method yields an amortized time complexity $T^\star$), then there exists a potential function $\Phi$ such that $\alg$ satisfies the Pyramid condition with respect to $\Phi$ and time~$T^\star$.
\end{theorem}

Theorem~\ref{thm:pushout_intro} states that every algorithm that has amortized delay $T^\star$ by the Push-out method can be analyzed by the Pyramid method, yielding the same bound $T^\star$ for the amortized delay.
At the same time, our algorithms for enumerating ideals and antichains in posets are examples of algorithms for which the Push-out method cannot be used to prove constant amortized delay, which in turn can be proved using the Pyramid method. 
We give an example poset in Example~\ref{ex:uno}.
This shows that the Pyramid method is strictly more powerful than the Push-out method.

The intuition is that in the Push-out method, every child $Y$ sends $\O{T(Y)}$ to its parent.
In contrast, the Pyramid method is a more versatile framework as it allows to transfere arbitrary amounts of computation time between children and parent iterations.
More precisely, it follows that the Push-out condition is satisfied if and only if Pyramid condition holds for a potential of the form $\Phi(X) = \O{T(X)}$. 
In this case, both methods can be applied to prove constant amortized delay.

	\subsubsection{Constant delay enumeration of poset ideals}
	The main result of this paper is a recursive enumeration algorithm for ideals in a poset that has constant delay and lists the ideals such that consecutively visited ideals differ in at most three elements. 
	We first describe a version that lists the ideals in arbitrary order with constant amortized delay. 
		A description in pseudocode is given in Algorithm~\ref{alg:ideal_intro}.
	Subsequently, we demonstrate how to refine this algorithm as to visit the ideals in Gray code order and invoke Theorem~\ref{thm:loopless_intro} to pass to constant delay.

	\begin{algorithm}[h]
		\caption{Enumeration of all ideals in a poset (basic version).}
		\label{alg:ideal_intro}
		\label{alg:ideal}
		\myproc{\Enumerate{Poset $P$, array $I$ of elements outside~$P$}}{
			\eIf{$P = \emptyset$}{
				\Visit{$I$}
			}{
				Find a longest chain $C$ in $P$\; \label{line:longestchain} 
				Construct auxiliary sets $L_2, \ldots, L_{k}$ and $S_0, \ldots, S_{k-1}$\;
				$P' \gets S_0$ \tcp{$P' = P_0$}\label{line:lisi}
				\Enumerate{$P'$, $I$} \tcp{Enumerate ideals in $P_0$}
				\For{$i \gets 1, \ldots, k$}{
					$P' \gets (P' \setminus L_i) \cup S_i$  \tcp{$P' = P_i$}\label{line:updatep}
					$I \gets I \cup L_i \cup \set{c_i}$\;\label{line:updatei}
					\Enumerate{$P'$, $I$} \tcp{Enumerate ideals in $P_i$}\label{line:rec}
				}
				$I \gets I \setminus P$ \tcp{Remove all elements added to $I$ by this iteration.}
			}
		}
		\myproc{\Main{Poset $P$}}{
			Relabel $P$ such that $u \preceq v$ implies $u \leq v$\;
			\Enumerate($P$, $\emptyset$)
		}
	\end{algorithm}
	
\paragraph{Recursive structure}
One main difference of our algorithm compared to the existing ones is the recursive structure.
In contrast to the previous algorithms, which split the set of ideals into two parts, depending on whether or not they contain a fixed element $x \in P$, we partition the ideals using a longest chain. 
This enables us to obtain constant delay if the input poset is a chain, which is not possible with the previously employed recursive techniques (see, for instance,~\cite{squire1995,abdo2013}).	
	
Let $C $ be a longest chain of $P$ and write $C = \{c_1, \dots, c_k\}$ with $c_1 \preceq \dots \preceq c_k$. 
The recursion of Algorithm~\ref{alg:ideal_intro} uses a partition of the ideals of $P$ into $k+1$ sets $\id_0(P), \dots, \id_k(P)$, depending on the largest element of $C$ contained in the ideal (or the fact that none exists). 
Formally, for $i > 0$, the set $\id_i(P)$ consists of the ideals $I$ of $P$ such that $c_i$ is the maximal element in $I \cap C$. 
The set $\id_0(P)$ contains those ideals of $P$ that do not contain an element of $C$.

\begin{figure}[htbp]
	\centering
\includegraphics[page=2]{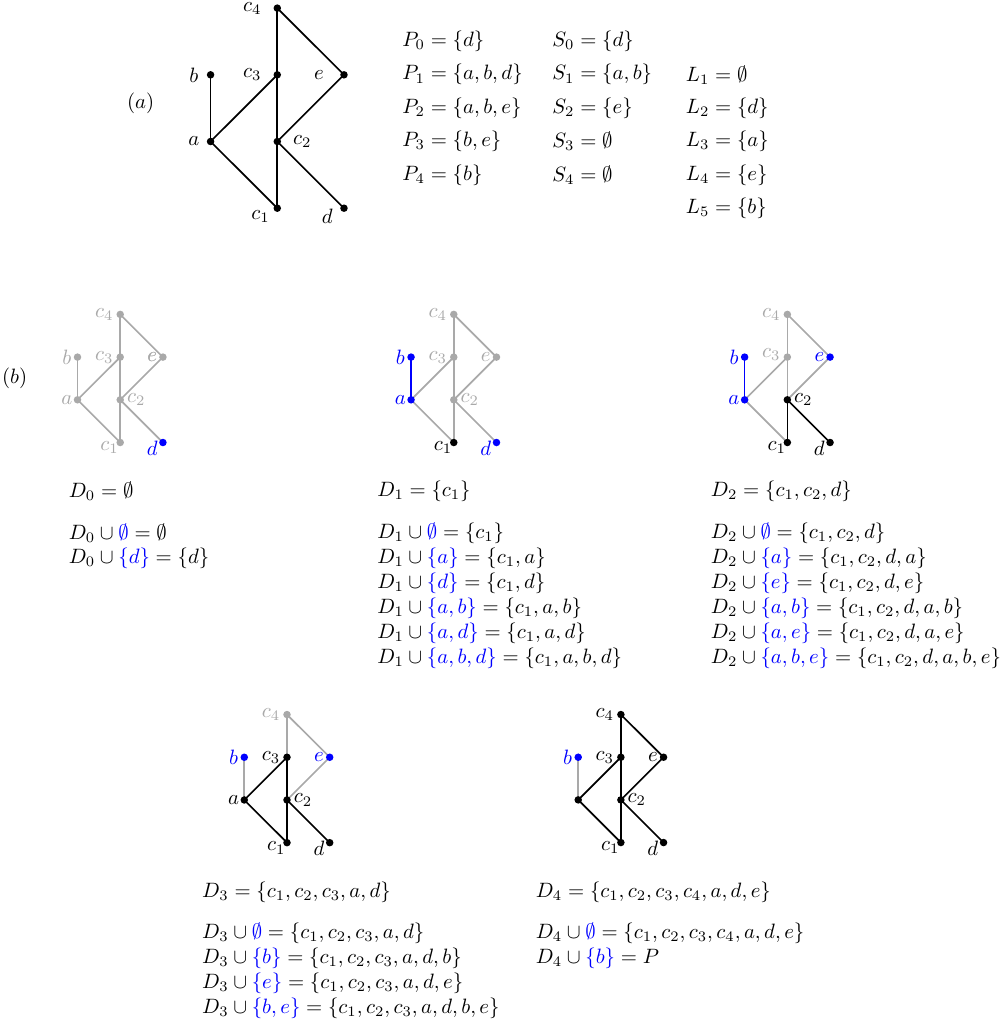}
\caption{(a) Example of a poset $P$, together with the subposets $P_0, \dots, P_4$ as well as the sets $S_0, \dots, S_4$ and $L_1, \dots, L_5$, visualized by its Hasse diagram. (b) The sets $D_i$ and the ideals in $\id_i(P)$ for $i = 0, \dots, 4$. 
Every ideal in $\id_i(P)$ is of the form $I' \cup D_i$, where $I'$ is an ideal of $P_i$. 
The elements in $P_i$ and $D_i$ are marked in blue and black, respectively.}
\label{fig:piexample}
\end{figure}

We now define the subposets $P_0, \dots, P_k$ of $P$ on which we call the recursion. 
For $i = 0, \dots, k$, the poset $P_i$ is the subposet of $P$ consisting of all elements $u \in P$ with $u \not \preceq c_i$ and $u \not \succeq c_{i+1}$ (using the convention that $c_0 \preceq u  \preceq c_{|C|+1}$ for all $u \in P$).
The crucial insight used for the recursion is that there is a bijection between the ideals in $\id_i(P)$ and the ideals of $P_i$ (see Lemma~\ref{lemma:idealcorrespondence}). 
It maps an ideal $I \in \id_i(P)$ to $I \setminus D_i$, where $D_i$ denotes the downset of $c_i$.
To obtain the ideals in $\id_i(P)$, we thus recursively compute the ideals in $P_i$ and add the set $D_i$ to each. 
Figure~\ref{fig:piexample} illustrates the definition of $P_i$ and the partition of the ideals into the sets $\id_i(P)$.

To efficiently update the currently processed poset when transitioning from $P_i$ to $P_{i+1}$ and compute the downset $D_i$ of $c_i$, we use auxiliary subsets $S_0, \dots, S_k$ and $L_1, \dots, L_{k+1}$ defined below.
These are computed simultaneously with the chain $C$.
We will now give details on these computations and their time complexity. 
	
\paragraph{Computing chains and updates} 
For computing a longest chain, we derive a special algorithm using an antichain decomposition in time $\O{n+q}$, where $q$ is the number of incomparable pairs, as the time complexity $\O{n^2}$ of the classical topological sorting algorithm is not sufficient for this task. 
Then, we compute sets $S_0, \dots, S_k$ and $L_1, \dots, L_{k+1}$ in time $\O{n+q}$. 
Here $S_i$ consists of all elements $u \in P \setminus C$ with $u \succ c_i$ and $u \not \succ c_{i+´1}$ (in other words, $c_i$ is the largest element on~$C$ that is smaller than $u$). 
Symmetrically, we define $L_i$ as the set of elements $u \in P \setminus C$ with $u \prec c_i$ and $u \not \prec c_{i-1}$.  
As $C$ is a longest chain, we have $S_k = L_1 = \emptyset$.
Note that with this notation, we have $D_i = D_{i-1} \cup L_i \cup \{c_i\}$ and $P_i = (P_{i-1} \setminus L_i) \cup S_i$.
We use these rules to update the current poset $P'$ on which the recursion is called, and the current set $I$ of elements that were already added to the ideal under construction.

Let $n_i = |P_i|$ for $i = 0, \dots, k$.
When processing the posets in the ordering $P_0, \dots, P_{k}$, we can show that every element of~$P_i$ is added and deleted a constant number of times, so for the updating, we require a total number of changes that is in $\O{\sum n_i} \leq \O{n+q}$. 
The key observation used for the complexity analysis of the construction of $C$ is that there is at most one comparison for every incomparable pair and every element is compared to at most one smaller element. 
A similar strategy is applied for computing all sets $S_i$, $L_i$, $D_i$ and $P_i$.
In order to analyze the complexity of whole algorithm, we apply the Pyramid method with a potential function of the form $\Phi(P) = \alpha + \beta n + \gamma q + \delta t$ with $\alpha, \beta, \gamma, \delta \in \mathbb{N}$. 
Here, $t$ denotes the number of incomparable triples (i.e., 3-antichains) in~$P$.

	
\paragraph{Gray code listings}	
The basic version of Algorithm~\ref{alg:ideal_intro} lists the ideals in arbitrary order. 
	We present a refined version that lists the ideals of the input poset $P$ such that every pair of consecutively visited ideals differs in at most three elements. 
	This is achieved by choosing a special ordering as well as prescribed start and end ideals for each of the recursive calls on $P_0, \dots, P_k$. 
	Using the technique from Theorem~\ref{thm:loopless_intro}, we can convert this into a loopless algorithm. 
	
	By a result of Pruesse and Ruskey~\cite{pruesse1993gray}, the flip graph on ideals of $P$ with flips between ideals differing in at most two elements has a Hamiltonian path.
	It is, however, not known whether such a path can be computed with constant amortized delay.
	We believe that it might be possible to refine our algorithm such that consecutively visited ideals differ in at most two elements, 
	while preserving constant delay.
	Note, however, that this does not influence the asymptotic runtime of the algorithm.
	
We summarize our results for enumerating poset ideals:
	
	\begin{theorem}
		Algorithm \ref{alg:ideal_intro} enumerates all ideals of a poset $P$ on $n$ elements with constant amortized delay while requiring $\O{n^2}$ space and $\O{n^2}$ time for the initialization.
		For $k \in \mathbb{N}$, the first $k$ ideals are visited in time $\O{k + n(n+q)}$, where $q \in \O{n^2}$ denotes the number of $2$-antichains in $P$. 
		
We can refine Algorithm~\ref{alg:ideal_intro} to a loopless algorithm that visits the ideals of $P$ in Gray code order, i.e, such that two consecutively visited ideals differ in at most three elements.
	\end{theorem}

	\subsubsection{Constant delay algorithms for enumerating antichains}
	Using the same techniques as for ideals, we obtain a constant delay enumeration algorithm for antichains (Algorithm~\ref{alg:antichains}). 
	As above, we can refine the algorithm to list the antichains in Gray code order, i.e, such that successively visited antichains differ in at most three elements.
	
			\begin{algorithm}[h]
		\caption{Enumeration of all antichains in a poset (basic version).}
		\label{alg:antichains}
		\myproc{\Enumerate{Poset $P$, array of elements $A$ outside $P$}}{
			\eIf{$P$ is empty}{
				\Visit{$A$}
			}{
				Find a longest chain $C$ in $P$\;\label{line:ac_chain}
				Construct $L_2, \ldots, L_{k}$ and $S_0, \ldots, S_{k-1}$\;\label{line:ac_lisi}
				$P' \gets S_0$ \tcp{$P' = P_1$}
				\Enumerate{$P'$, $A \cup \set{c_1}$}\;\label{line:call1}
				\For{$i \gets 2, \ldots, k$}{
					$P' \gets (P' \setminus L_i) \cup S_{i-1}$ \tcp{$P' = P_i$}
					\Enumerate{$P'$, $A \cup \set{c_i}$}\;\label{line:call2}
				}
				\Enumerate{$P \setminus C$, $A$}\;
			}
		}
		\myproc{\Main{Poset $P$}}{
			Sort the elements so that $u \preceq v$ implies $u \leq v$\;
			\Enumerate($P$, $\emptyset$)
		}
	\end{algorithm}

		\begin{theorem}
		Algorithm \ref{alg:antichains} enumerates all antichains of a poset $P$ on $n$ elements with constant amortized delay, and requires $\O{n^2}$ space and $\O{n^2}$ time for the initialization, where $q \in \O{n^2}$ denotes the number of $2$-antichains in $P$.
		The first $k$ antichains are visited in time $\O{k + n(n+q)}$.
		We can refine Algorithm~\ref{alg:antichains} to a loopless algorithm that visits the antichains of $P$ in Gray code order, i.e, such that two consecutively visited antichains differ in at most three elements.
	\end{theorem}

This paper is organized as follows: in Section~\ref{sec:potentialanalysis}, we describe the Pyramid condition, a new method for the potential analysis of recursive algorithms. 
In Section~\ref{sec:ideals}, we describe an algorithm that enumerates ideals in a poset with constant delay. 
In Section~\ref{sec:antichains}, we give a constant delay algorithm for antichains.

		\section{Potential analysis for combinatorial enumeration}\label{sec:potentialanalysis}
		
		In this section, we describe a new method for potential analysis, called \emph{Pyramid method}. 
		It will be used to show that the algorithms for enumerating ideals and antichains in posets (Algorithms~\ref{alg:ideal} and~\ref{alg:antichains}) have constant amortized delay. 
		We stress, however, that the presented methods can be applied to any enumeration algorithm, and might thus be of independent interest. 
		
		In Section~\ref{sec:recursive_setup}, we introduce the setup of the analyzed algorithms. 
		In Section~\ref{sec:pyramid}, we describe the Pyramid method. 
		In Section~\ref{sec:worstcase}, we discuss how turn an algorithm that satisfies the Pyramid condition with amortized delay $T^\ast$ into an algorithm with worst-case delay $T^\ast$, provided that the elements are visited in Gray code order.

		\subsection{Recursive algorithms}\label{sec:recursive_setup}
We now describe the setup of the recursive algorithms that we analyze using terminology similar to Uno~\cite{uno2015constant}. 
Let $S$ be a set to be enumerated using a recursive algorithm $\alg$. 
As an example, consider the following example enumerating all matchings of a graph. 

	\begin{figure}
	\centering
	\includegraphics{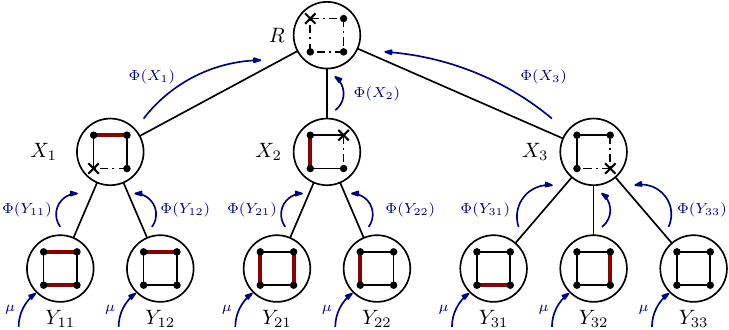}
	\caption{Recursion tree for enumerating all matchings in a graph $G$ (see Example~\ref{ex:rec}).
		In each iteration (node of the tree), a vertex $v$ of $G$ (marked with $\times$) is selected. 
		The recursive calls correspond to selecting an edge incident to $v$ to be included in the matching (marked in red) or deciding to include none of the edges incident to $v$ (excluded edges are marked in solid black).
		At the beginning, the status of all edges in $G$ is undetermined (dashed). Every leaf of the recursion tree corresponds to a matching of $G$.
	}
	\label{fig:recursiontree}
\end{figure}

	\begin{example}\label{ex:rec}
	We consider a simple algorithm to enumerate all matchings in an input graph~$G$ (also see~\cite[Section 5]{uno2015constant}). 
	For an illustration, see Figure~\ref{fig:recursiontree}. 
	The algorithm maintains a set of edges included in the matching. At the beginning, this set is empty.
	In each iteration, we fix a vertex $v$ of $G$ and let $e_1, \dots, e_k$ be the edges incident to $v$. 
	Clearly, every matching of $G$ contains at most one of these edges.
	Thus, for each $i \in \{1, \dots, k\}$, there is a recursive call of the function that includes the edge $e_i$ in the matching, together with one additional call representing the case that none of the edges $e_1, \dots, e_k$ is included. 
\end{example}
		
%

Let $T(\alg)$ be a recursion tree representing the recursion calls in $\alg$.
Every node of $T(\alg)$ corresponds to a recursive function call of $\alg$ and potentially some additional non-recursive computation.
		An \emph{iteration} of the recursive function $\alg$ is its execution, excluding the computation in the function recursively called by the iteration.
		We can identify the iterations of $\alg$ with the nodes of its recursion tree, and use them interchangeably. 
		Associated to each iteration is a subset $X \subseteq S$ that contains the elements that are visited in this iteration or one of its recursive calls.
		We use the set~$X$ to label and identify the iteration.
		In iteration $X$, a partition $X = X_1 \dcup \cdots \dcup X_k \dcup X'$ is derived, where $X_1, \dots, X_k \neq \emptyset$.
		Here,~$X'$ denotes the set of elements directly visited by the iteration $X$, and may be empty. 
		The function is recursively called for the sets $X_1, \ldots, X_k$. 
		These iterations are called the \emph{children} of $X$ as in the recursion tree, the nodes corresponding to $X_1, \dots, X_k$ are the children of the node corresponding to $X$.
		The set of children of $X$ is denoted by $C(X)$, and we call~$X$ the \emph{parent} of $X_1, \dots, X_k$.		
		The \emph{root iteration} is the iteration without parent, and a \emph{leaf iteration} is an iteration without any children.
		A \emph{subtree} of an iteration $X$ contains $X$ and all its descendants, i.e., iterated children sets.

		In our applications, we usually encounter the special case that if $|X| \geq 2$, no element is directly visited by the iteration $X$ and all elements of $X$ are passed on to the recursive calls, and if $|X| =1$, the unique element of $X$ is visited and no recursion is called.
		That is, we have $X' = \emptyset$ if $|X| > 1$ and $X' = X$ if $|X| =1$.  

		As customary in sublinear enumeration, we assume that every visited element is provided by a pointer to a data structure (usually an array) and this data structure is modified during the enumeration before another element is visited.

		\subsection{Potential analysis and the Pyramid condition}\label{sec:pyramid}

		In this section, we describe a potential analysis method which we call \emph{Pyramid method}.
		
		The potential analysis is a well-established technique to analyze the amortized time complexity of data structures (for instance, see~\cite{cormen2009}).
		We now define a variant to analyze the amortized complexity of recursive algorithms, in particular for combinatorial enumeration. 
		The intuition is that for many enumeration algorithms, most elements of the enumerated set are visited in leaf iterations. 
		At the same time, these iterations tend to require only small computation time. 
		In other words, the leaf iterations usually visit elements in quick succession while requiring little computation time, whereas the inner nodes of the tree visit less elements and require more computation time. 
		Leaf iterations thus tend to have a positive balance between used computation time and generated output, whereas inner nodes have a negative one. 
		For the amortized time analysis, this suggests a model in which children nodes pass excess time resources to their parent. 
		This approach was, for instance, taken by Uno's Push-out method~\cite{uno2015constant}, which allowed to show that various enumeration algorithms have constant delay.
		
		In our model, we use \emph{coins} as fixed time units. 
		Each coin pays for $T^\star$ instructions, where $T^\star$ can be any fixed function of the input size. 
		Ultimately, $T^\star$ will be the amortized delay that our analysis yields. 
		In our applications, we usually show that our algorithm has constant delay, which corresponds to setting $T^\star = \O{1}$.
		If an iteration visits an element, it is rewarded by obtaining a fixed constant number $\mu > 0$ of coins. 
		When executing instructions, the corresponding amount of coins needs to be payed by the algorithm. 
		
		Following the above intuition, our model makes children iterations pass coins to their parents. 
		This is modeled by a potential function $\Phi$. 
		It assigns each iteration $X$, i.e., each node of the recursion tree, a potential $\Phi(X) >0$ representing the number of coins that~$X$ passes to its parent.
		An iteration $X$ thus obtains $\mu$ coins for every visited element of $X'$ as well as $\sum_{Y \in C(X)} \Phi(Y)$ coins from its children, and passes $\Phi(X)$ coins to its parent (see Figure~\ref{fig:recursiontree} for an illustration for the algorithm from Example~\ref{ex:rec}).
		
		Intuitively, this mechanism is sustainable if $\Phi$ and $T^\star$ are chosen in a way that every node can pay its own computation time.
		That is, the difference between the $\mu |X'| + \sum_{Y \in C(X)} \Phi(Y)$ coins that an iteration $X$ obtains and the $\Phi(X)$ coins it sends to its parent must be sufficient to 
		pay for the execution of the instructions in $X$, i.e, at least $T(X)/T^\star$ coins are needed.
		This is captured by the following condition:
		
		\begin{definition}[Pyramid condition]\label{def:pyramidcondition}
A recursive algorithm $\alg$ satisfies the \emph{Pyramid condition} with respect to a potential function $\Phi$ and time $T^\star$ if for every iteration $X$ of~$\alg$, it holds that \[\sum_{Y \in C(X)} \Phi(Y) + \mu |X'| - \Phi(X) \ge \frac{T(X)}{T^\star}.\]
		\end{definition}
	
		Note that this requires the root iteration $R$ to provide $\Phi(R)$ coins to pass to its (non-existing) parent. 
		This additional time resource could be used to cover instructions executed outside the recursion, such as reading the input or initialization.
	
		We now show that this condition can indeed be used to obtain bounds for the amortized delay of a recursive algorithm. 
		
		\subsubsection{Amortized delay}
		\begin{lemma}\label{lem:subtree_time}
		Let $\alg$ be a recursive algorithm enumerating a set $X$ that satisfies the Pyramid condition for some potential function $\Phi$ and time $T^\star$. 
		Let $X$ be an iteration of $\alg$. 	
		Then the number of instructions executed in a subtree of $X$ is at most $T^\star (\mu |X| - \Phi(X))$.
		\end{lemma}
		\begin{proof}
			We prove the lemma by the induction on the size of $X$.
			The number of instructions executed in the subtree of $X$ is at most
			\begin{align*}
				&T(X) + \sum_{Y \in C(X)} T^\star(\mu|Y| - \Phi(Y)) \\
				&= T^\star \left(\frac{T(X)}{T^\star} +  \Phi(X) - \mu|X'| - \sum_{Y \in C(X)} \Phi(Y) \right) + T^\star(\mu |X| - \Phi(X)) \\
				&\le T^\star(\mu |X| - \Phi(X)).
			\end{align*}
		Note that this inequality holds both when $X$ is a leaf, i.e, $C(X) = \emptyset$, as well as when $C(X)$ is non-empty.
		In the second step, we used that $|X| = |X'| +\sum_{Y \in C(X)} |Y|$.
		The last inequality follows from the Pyramid condition.
		\end{proof}
		
		This immediately yields the following:
		
		\begin{theorem} \label{thm:amortized}
		Let $\alg$ be a recursive algorithm enumerating a set $X$  and assume that $\alg$ satisfies the Pyramid condition for some potential function $\Phi$ and time $T^\star$. 
		The total time to enumerate $X$ is $\O{|X| T^\star - \Phi(X)}$. 
		Thus $\alg$ has amortized delay $\O{T^\star}$.
		\end{theorem}
		We stress that this only takes the computations executed within the recursion into account, in particular not reading the input and possible initialization.

		\subsubsection{Visiting times}
		An enumeration algorithm defines an ordering on the enumerated set $X$, namely the ordering in which it visits the elements of $X$. 
		We now analyze the number of instructions that are executed between visiting the $i$-th and the $j$-th element of this ordering.
		Let $\Delta$ be an upper bound for the number of instructions executed by all iterations on a path in the recursion tree from the root to a leaf iteration.
		
		\begin{lemma} \label{lem:gap}
			Let $\alg$ be a recursive algorithm enumerating a set $X$  and assume that $\alg$ satisfies the Pyramid condition with time $T^\star$. 
			For every $1 \le i \le j \le |X|$, the number of instructions $\delta_{ij}$ executed between visiting the $i$-th element and the $j$-th element of $X$ is at most $2\Delta + (j-i) \mu T^\star-1$.
			Furthermore, the number of instructions executed before visiting the $j$-th element is at most $\Delta + j \mu T^\star-1$.
		\end{lemma}

		\begin{proof}
			For an illustration, see Figure~\ref{fig:gap}.
			\begin{figure}
			\centering
			\includegraphics{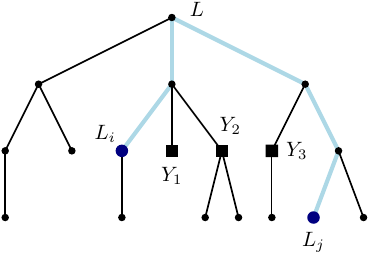}
			\caption{Illustration of the proof of Lemma~\ref{lem:gap}.}
			\label{fig:gap}
			\end{figure}
			Let $L_i$ and $L_j$ be the iterations in which the $i$-th and the $j$-th element are visited, respectively. 
			Let $L$ be the lowest common ancestor of $L_i$ and $L_j$ in the recursion tree.
			Let $Y_1, \dots, Y_k$ denote those iterations executed between $L_i$ and $L_j$ whose parents are on the paths from $L$ to $L_i$ and $L_j$, but that are not contained in the paths themselves.
			Then $\delta_{ij}$ is bounded from above by the number of instructions on the paths, together with the instructions in the subtrees of $Y_1, \dots, Y_k$. 
			The number of instructions on the paths is at most $2 \Delta - T(L) \leq 2\Delta-1$.
			By Lemma~\ref{lem:subtree_time}, the number of instructions executed in the subtree of iteration $Y_i$ is at most $\mu T^\star |Y_i|$.
			Note that the sets $Y_1, \dots, Y_k$ contain all elements visited between $i$ and $j$ and that each element is contained in precisely one set $Y_i$. 
			Thus $\sum_{i=1}^k  \mu T^\star |Y_i| = \mu T^\star (j-i)$.
			This yields $\delta_{ij} \leq 2 \Delta +\mu T^\star (j-i)-1$. 
			
			The proof of the second part follows analogously by summing the instructions executed along a path from the root to $L_j$ as well as in all subtrees whose root has a parent on this path and which are executed before $L_j$ (note that the $j$-th element is not contained in any of these subtrees).
		\end{proof}
	
			This immediately yields the following result:
		
		\begin{theorem} \label{thm:kth}
				Let $\alg$ be a recursive algorithm enumerating a set $X$ and assume that $\alg$ satisfies the Pyramid condition with time $T^\star$. 
				Then $\alg$ enumerates the first $k$ elements in time $\O{k T^\star + \Delta}$.
		\end{theorem}
	
	Note that this yields an upper bound on the amortized delay defined in the sense of Cappelli and Strozecki~\cite{capelli2021amortized}.
	
			\subsubsection{Comparison with Uno's Push-out method}\label{sec:pushout}
	
	Uno \cite{uno2015constant} presented a method called \emph{Push-out} to analyze the amortized complexity of recursive enumeration algorithms.
	It uses the following condition:
	
	\begin{definition}[Push-out condition~\cite{uno2015constant}]\label{def:pushout}
A recursive algorithm satisfies the \emph{Push-out condition} with time $T^\star$ if there exist constants $\alpha > 1$ and $\beta \geq 0$ such that for all iterations~$X$, we have
	\[\sum_{Y \in C(X)} T(Y) \ge \alpha T(X) - \beta (|C(X)|+1) T^\star.\]
	\end{definition}

\begin{theorem}[{\cite[Theorem~1]{uno2015constant}}]\label{thm:uno}
A recursive enumeration algorithm that satisfies the Push-out condition with time $T^\star$ has amortized delay~$\O{T^\star}$.
\end{theorem}

	We prove that the Pyramid method subsumes the Push-out method by defining a particular potential function for which both methods give equivalent conditions.
	
	\begin{proposition}\label{obs:push_out}
If a recursive algorithm $\alg$ satisfies the Push-out condition with time $T^\star$ for constants $\alpha > 1$ and $\beta \geq 0$, the $\alg$ satisfies the Pyramid condition with time $T^\star$ for the potential function 
$$\Phi(X) = \frac{T(X)}{(\alpha-1)T^\star} + \frac{3\beta}{\alpha-1}.$$
	\end{proposition}

\begin{proof}
	We have
	\begin{align*}
		\sum_{Y \in C(X)} \Phi(Y) - \Phi(X)
		=&\; \frac{1}{(\alpha-1)T^\star}\left(\sum_{Y \in C(X)}T(Y)-T(X)\right) + \frac{3\beta}{\alpha-1} (|C(X)|-1)\\
		\ge&\; \frac{1}{(\alpha-1)T^\star}\left((\alpha-1)T(X) - \beta(|C(X)|+1)T^\star \right) + \frac{\beta}{\alpha-1} (3|C(X)|-3)\\
		=&\; \frac{T(X)}{T^\star} + \frac{\beta}{\alpha-1}(2|C(X)|-4) \\
		\ge&\; \frac{T(X)}{T^\star}.
	\end{align*}
The first inequality uses the Push-out condition and the second one follows from the fact that every inner node of the recursion tree has at least two children.
\end{proof}

	This result, combined with Theorem~\ref{thm:amortized}, yields an alternative proof of \cite[Theorem~1]{uno2015constant}. 
	Moreover, it shows that any enumeration problem analyzed by the Push-out method can also be analyzed by the Pyramid method, yielding the same amortized time complexity. 
	In particular, we can invoke Lemma~\ref{lem:gap} to determine the complexity of enumerating the first $k$ elements.
	This applies, for instance, to the enumeration of elimination orderings, matchings forests, connected induced subgraphs, and spanning trees in a given graph studied in~\cite{uno2015constant}.
	
	In Example~\ref{ex:uno} below, we show that for Algorithm~\ref{alg:ideal} enumerating ideals in a poset, the Push-out method fails to show that the given algorithm has constant delay, whereas we can show this with the Pyramid method. 
	The same applies for Algorithm~\ref{alg:antichains}.

		\subsection{Worst-case delay for Gray code algorithms}\label{sec:worstcase}
		
		We now describe a method to deamortize a recursive algorithm $\alg$ that satisfies the Pyramid condition with time $T^\star$, i.e.,~$\alg$ has amortized delay $T^\star$, provided that $\alg$ enumerates the elements in so-called \emph{Gray code order}.
		The goal is to obtain an algorithm~$\halg$ with worst-case delay $\O{T^\star}$ that enumerates the elements in the same ordering as $\alg$.
		This method was already described in an unpublished technical report by Uno~\cite{uno2003}. 
		We include a description here for completeness.
		Capelli and Strozecki~\cite{capelli2021amortized} present a different approach for deamortization.

		We assume that $\alg$ enumerates all elements in a \emph{Gray code order}. 
		This means that the modifications in the data structures between two consecutively visited elements can be stored in $\O{1}$ memory, and the data structure can be modified for the next element in the worst-case time $\O{T^\star}$.
		
		The deamortization of $\alg$ uses an auxiliary queue $Q$ with $\frac{2\Delta}{\mu T^\star}$ slots.
		Every slot stores the changes between two consecutively visited elements.
		Algorithm $\halg$ simulates the computation of Algorithm~$\alg$ as follows:
		whenever $\alg$ visits an element, Algorithm~$\halg$ stores it in the queue $Q$.
		More precisely, $\halg$ saves the changes in the data structure between two elements consecutively visited by~$\alg$, thereby using one slot of $Q$.
		
		In the beginning, $\Delta$ instructions of $\alg$ are executed.
		Subsequently, after each $\mu T^\star$ instructions of Algorithm~$\alg$, Algorithm $\halg$ visits and removes the first element in $Q$.
		Concretely, $\halg$ applies all changes stored in the first slot of $Q$ to the data structure. 
		By assumption, this can be computed in time $\O{T^\star}$.
		Furthermore, whenever $Q$ contains more than $\frac{2\Delta}{\mu T^\star}$ elements, Algorithm~$\halg$ visits and removes the first element in $Q$.
		It remains to prove the following:
		
		\begin{lemma}
		The queue $Q$ is non-empty whenever Algorithm $\halg$ needs to visit an element.
		\end{lemma}
	
		\begin{proof}
			Assume that $Q$ is empty when Algorithm $\halg$ needs to visit the $j$-th element.
			First assume that~$Q$ has never been full before.
			When $\halg$ needs to visit the $j$-th element, Algorithm~$\alg$ has executed $j \mu T^\star + \Delta$ instructions.
			By Lemma \ref{lem:gap}, Algorithm $\alg$ executed at most $j \mu T^\star + \Delta-1$ instructions before visiting the $j$-th element.
			This contradicts the assumption that~$Q$ is empty.
			
			Now assume that the last occasion when $Q$ was full was when $\alg$ visited the $i$-th element.
			From the size of the queue, it follows that $\alg$ executed $(j-i) \mu  T^\star + 2\Delta$ instructions.
			However, by Lemma~\ref{lem:gap}, Algorithm~$\alg$ executed at most $(j-i) \mu  T^\star + 2\Delta-1$ instructions between visiting the $i$-th and the $j$-th element.
			This is a contradiction.
		\end{proof}
		
		This yields the following theorem for deamortization:
		
		\begin{theorem}\label{thm:loopless}
		Let $\alg$ be a recursive algorithm satisfying the Pyramid condition with time~$T^\star$. 
		Furthermore, assume that the change between two successively visited elements of $\alg$ can be stored using $\O{1}$ memory and that these changes can be applied in time $\O{T^\star}$. 
		Then Algorithm $\halg$ enumerates all elements in the same order as $\alg$ with worst-case delay $\O{T^\star}$ while requiring $\O{\Delta/T^\star}$ additional space and $\O{\Delta}$ additional time for the initialization.
		\end{theorem}

	\section{Constant amortized time enumeration of poset ideals}\label{sec:ideals}
	
	In this section, we present a recursive algorithm (Algorithm~\ref{alg:ideal})  that enumerates all ideals of an input poset $(P,\preceq)$ with constant amortized delay.
	
	\subsection{Preliminaries}\label{sec:ideal_prelims}
	We first introduce some auxiliary structures that will be used to define the recursive structure of our algorithm. 
	Throughout, let $P$ be a poset and fix a longest chain $C$ of $P$.
	Let $k =|C|$. 
	The recursion of the enumeration algorithm uses a partition of the ideals of $P$ into $k+1$ groups, depending on the maximal element of $C$ contained in the ideal or the fact that the ideal does not contain any element of $C$, respectively.
	
	For an illustration of the following definitions, see Figure~\ref{fig:piexample}.
	Let $\id(P)$ denote the set of ideals of~$P$.
	Write $C = \{c_1, \dots, c_{k}\}$ with $c_1 \prec \dots \prec c_{k}$.
	For simplicity, we add two virtual elements $c_0$ and $c_{k+1}$, which are smaller and larger than all elements in $P$, respectively.
	For $I \in \id(P)$, let $\zeta(I) = \max\{i \colon c_i \in I\}$ and set $\zeta(I) = 0$ if $I \cap C = \emptyset$. 
	For $i \in \{0, \dots, k\}$, let $\id_i(P) = \{I \in \id(P) \colon \zeta(I) = i\}$.
	We set $D_i = \{u \in P \colon u \preceq c_i\}$ to be the downset of $c_i$, and $U_i = \{u \in P \colon u \not \succeq c_{i+1}\}$. 
	Note that $D_i$ and $U_i$ are the smallest and largest ideal in $\id_i(P)$, respectively. 
	In particular, we have $D_0 = \emptyset$ and $U_{k} = P$.
	
	For any $u \in P$, let
	\[s_u = \max\{i \in \{0, \dots, k\} \colon c_i \preceq u\}\text{ and } l_u = \min\{i \in \{0, \dots, k\} \colon u \preceq c_i\} \]
	denote the indices of the smallest element of~$C$ smaller than $u$ and the largest element of~$C$ larger than $u$, respectively.
	Note that $s_{c_i} = l_{c_i} = i$.
	For $i \in \set{0, \ldots, k+1}$, let 
	\[S_i = \set{u \in P \setminus C \colon s_u = i} \text{ and } L_i = \set{u \in P \setminus C \colon l_u = i}.\]
	Note that $L_1 = S_{k} = \emptyset$.
	Moreover, we have $D_i \setminus C = L_2 \cup \cdots \cup L_i$ and $U_i \setminus C = S_{1} \cup \cdots \cup S_{i}$. 
	Let $P_i = (P \setminus D_i) \cap U_i$. 
	In other words, $P_i$ consists of those elements in $u \in P$ with $u \not \preceq c_i$ and $u \not \succeq c_{i+1}$. 
	In particular, the following is easily verified:
	
	\begin{lemma} \label{lem:ideal-equiv}
		For $u \in P \setminus C$ and $i \in \set{0, \ldots, k}$, the following properties are equivalent: 
		\begin{enumerate}
		\item $u \in P_i$,
		\item $s_u \le  l_u$,
		\item $u$ is incomparable with $c_i$ or $c_{i+1}$.	
		\end{enumerate}
	\end{lemma}

In particular, $P_0$ and $P_{k}$ consist of all elements of $P$ which are incomparable to $c_1$ and to~$c_{k}$, respectively.
From Lemma~\ref{lem:ideal-equiv}, the following easily follows:

\begin{lemma}\label{lem:once}
	Every antichain~$A$ of $P \setminus C$ is contained in at least one of the posets $P_0, \ldots, P_{k}$.
	Furthermore, $A$ lies in at least two such posets if and only if $C$ contains an element incomparable to all elements of~$A$.
\end{lemma}

\begin{proof}
Choose $u \in A$ such that $s_u$ is maximal. 
By Lemma~\ref{lem:ideal-equiv}, we have $u \in P_{s_u}$. 
If $A \not \subseteq P_{s_u}$, then there exists $v \in A \setminus  P_{s_u}$. 
By Lemma~\ref{lem:ideal-equiv} and the maximality assumption on $u$, this implies $l_v \leq s_u$. 
But then $v \preceq c_{l_v} \preceq c_{s_u} \preceq u$, which is a contradiction to $u$ and $v$ being incomparable. 
The second statement can be shown similarly.
\end{proof}

This lemma implies that every element of $P \setminus C$ lies in at least two of the posets $P_0, \ldots, P_{k}$.
In particular, every element of $L_i$ is incomparable with $c_{i-1}$ and every element of $S_i$ is incomparable with $c_{i+1}$.	
	
We now decompose the set of ideals of $P$. 
Clearly $\id(P) = \id_0(P) \dcup \cdots \dcup \id_{k}(P)$. 
	The next lemma establishes a bijective correspondence between $\id_i(P)$ and $\id(P_i)$, which is the central ingredient for the recursion. 
	
	\begin{lemma}\label{lemma:idealcorrespondence}
		For $i \in \{0, \dots, k\}$, the map 
		\[\varphi_i^P \colon \id_i(P) \to \id(P_i),\,I \mapsto I \setminus D_i\] 
		is a (well-defined) bijection. 	
		Its inverse maps $I' \in \id(P_i)$ to $I' \cup D_i \in \id_i(P)$.
	\end{lemma}
	
	\begin{proof}
Let $I \in \id_i(P)$. 
As $I$ is downward-closed, we have $D_i \subseteq I$. 
Moreover, as $c_{i+1} \notin I$, we have $I \cap (S_{i+1} \cup \dots \cup S_{k-1}) = \emptyset$. 
Thus $I \setminus D_i$ is an ideal of $P_i$.
Now let $I' \in \id(P_i)$. 
Then clearly $I' \cup D_i$ is downward-closed, so an ideal of $P$.

For $I \in \id_i(P)$, we have $(I \setminus D_i) \cup D_i = I$ as $D_i \subseteq I$. 
Conversely, for $I' \in \id(P_i)$, we have $(I' \cup D_i)\setminus D_i = I'$ as $I' \cap D_i \subseteq P_i \cap D_i = \emptyset$. 
In particular, $\varphi_i^P$ is a bijection.
\end{proof}

We conclude this part with some observations that will be used in Section~\ref{sec:graycodeideals} to define a Gray code for $\id(P)$. 	
	
	\begin{lemma}\label{lemma:propertiesdi}
Let $i \in \{0, \dots, k-2\}$. 
		\begin{enumerate}
			\item For $I \in \id_i(P)$, the set $I \cup \set{c_{i+1},c_{i+2}}$ is an ideal of $P$ if and only if $L_{i+1} \cup L_{i+2} \subseteq I$.
Therefore, $U_i \cup \set{c_{i+1},c_{i+2}} \in \id_{i+2}(P)$ and $U_i \setminus D_{i+2} \in \id(P_{i+2})$.
			\item Similarly, for $I \in \id_{i+2}(P)$, the set $I \setminus \set{c_{i+1},c_{i+2}}$ is an ideal of $P$ if and only if $(S_{i+1} \cup S_{i+2}) \cap I = \emptyset$.
			Therefore, $D_{i+2} \setminus \set{c_{i+1},c_{i+2}} \in \id_{i}(P)$ and $D_{i+2} \setminus (D_{i} \cup \set{c_{i+1},c_{i+2}}) \in \id(P_{i})$.
		\end{enumerate}
	\end{lemma}
	
\begin{proof}
We only prove the first claim, the second is analogous. 
Let $I \in \id_i(P)$. 
First assume that $I \cup \set{c_{i+1},c_{i+2}}$ is an ideal of $P$. 
Then $L_{i+1} \cup L_{i+2} \subseteq D_{i+1} \cup D_{i+2} \subseteq I \cup \{c_{i+1}, c_{i+2}\}$ as this set is downward-closed. 
Conversely, assume $L_{i+1} \cup L_{i+2} \subseteq I$. 
We have $D_{i+2} = D_i \cup L_{i+1} \cup L_{i+2} \subseteq I$ as $D_i \subseteq I$ due to $I \in \id_i(P)$, and hence $I \cup \{c_{i+1}, c_{i+2}\}$ is downward-closed.
In particular, this yields $U_i \cup \set{c_{i+1},c_{i+2}} \in \id_{i+2}(P)$. 
By Lemma~\ref{lemma:idealcorrespondence}, we then have $\varphi_{i+2}^P(U_i \cup \{c_{i+1}, c_{i+2}\}) = U_i \setminus D_{i+2} \in \id(P_{i+2})$.
\end{proof}

\subsection{Constant amortized time algorithm}

Throughout, let $P$ be a poset on $n$ elements. 
We now describe an algorithm that lists all ideals of~$P$ with constant amortized delay. 
In this section, we describe a basic version (Algorithm~\ref{alg:ideal}) that does not necessarily list the ideals in Gray code order. 
In the subsequent section, we show how this version can be modified to output a Gray code listing by selecting a specific order for the recursive calls. 	

We use the notation from the Section~\ref{sec:ideal_prelims}. 
Additionally, let $Q$ denote the set of incomparable pairs of $P$ and set $q= |Q|$.

\paragraph{Description of the algorithm} 	
	

We use the bijection between $\id(P)$ and $\id(P_0) \dcup \cdots \dcup \id(P_{k})$ described in Lemma~\ref{lemma:idealcorrespondence}.
Given an input poset $P$, an iteration of Algorithm~\ref{alg:ideal} consists of computing the posets $P_0, \dots, P_{k}$ and proceeding recursively on these smaller posets.
Throughout the recursion, a list of elements $I$ containing the elements that have been already added to the currently processed ideal by a parent iteration is maintained. 
By Lemma~\ref{lemma:idealcorrespondence}, we obtain an ideal of $P$ from an ideal of $P_i$ by adding the elements in $D_i$. 
Thus the set $D_i = D_{i-1} \cup L_i \cup \{c_i\}$ is added to the array~$I$ at the beginning of the recursive call for the poset $P_i$. 

\paragraph{Implementation}

We now describe the implementation of the steps of Algorithm~\ref{alg:ideal} as well as the required data structures in detail. 
\medskip
	
\emph{Input and output:} The input poset $P$ is represented by an adjacency matrix $M$ with $M_{uv} = 1$ if $u \prec v$ and $M_{uv} = 0$ otherwise.
In the initialization phase, we relabel the elements of~$P$ to guarantee that $u \preceq v$ implies $u \leq v$. 
The ideals visited by Algorithm~\ref{alg:ideal} are stored in an array $I$.
\medskip

\emph{Computation of a longest chain.} We now describe the computation of a longest chain (Line~\ref{line:ac_chain}). 
In order to achieve that Algorithm~\ref{alg:ideal} runs with constant amortized delay, this computation needs to have time complexity $\O{n+q}$. 
In particular, we cannot use the topological sorting algorithm as it requires $\Omega(n^2)$ time.

Instead, we first compute an antichain decomposition of $P$, using the following algorithm (see Algorithm~\ref{alg:antichaindecomp}):
\begin{algorithm}[h]
	\caption{Creating an antichain decomposition.}
	\label{alg:antichaindecomp}
	\myproc{\Main{Poset $P$ on elements $\{1, \dots, n\}$ such that $u \preceq v$ implies $u \leq v$}}{
		$j \gets 0$\;
		\For{$i \gets 1, \ldots, n$}{	
			\For{$t \gets j, \ldots, 1$}{
				\For{$v \in A_t$}{
					\If{$v \prec u$}{
						\eIf{$t =j$}{
							$A_{j+1} \gets \{u\}$ \tcp{create a new antichain $A_{j+1}$}
							$x_u = v$\;
							$j \gets j+1$\;
							\KwGoTo \texttt{iteration\_end} \tcp{break two inner loops}
						}
						{
							$A_{j+1} \gets A_{j+1} \cup \{u\}$ \tcp{add element to the existing antichain $A_{j+1}$}
							$x_u = v$\;
							\KwGoTo \texttt{iteration\_end}\;
						}
					}
					
				} 
				$A_{1} \gets A_{1} \cup \{u\}$ \tcp{no smaller element found, so add $u$ to the bottom antichain $A_{1}$}
				$x_u = \bot$\;
			}
			\texttt{iteration\_end:}
		}
		\Return{$A_1, \dots, A_j$}
	}
\end{algorithm}
We process the elements of $P$ individually in the order $1, \dots, n$.
Throughout, we maintain a list of antichains $A_1, \ldots, A_j$ as well as a variable $x_u$ for every previously processed element $u$.
A new element $u$ is inserted by successively comparing it to the elements in the antichains $A_j, \ldots, A_1$ until the first element $v \in A_i$ with $v \prec u$ is found. 
In this case we set $x_u = v$ and add $u$ to $A_{i+1}$. 
This may cause the creation of a new antichain $A_{j+1}$.
If $u$ is incomparable to all previously processed elements, $u$ is added to $A_1$, and we set $x_u = \bot$ to be the null pointer.
This yields an antichain decomposition $A_1, \dots, A_k$ of $P$.

We then find a longest chain $C$ in $P$ by starting from an arbitrary element $u \in A_k$ and following the variables $x_u$.
This clearly creates a chain~$C$ of size~$k$.
As every chain contains at most one element from each antichain $A_1, \dots, A_k$, $C$ is a longest chain in $P$.

The correctness of this algorithm follows from the following lemma:

\begin{lemma}\label{lem:antichain_correctness}
	For every $i \in \{1, \dots, k\}$, the set $A_i$ computed by Algorithm~\ref{alg:antichaindecomp} is an antichain.
	Algorithm~\ref{alg:antichaindecomp} and thus the computation of a longest chain can be carried out in $\O{n+q}$.
\end{lemma}

\begin{proof}
	Suppose that there exist $i \in \{1, \dots, k\}$ and $u,v \in A_i$ with $u \prec v$. Then $u \leq v$, so~$v$ was inserted after $u$. 
	Thus $v$ would have been compared to $u$ and thus inserted in $A_{i+1}$, which is a contradiction. 
	
	Algorithm~\ref{alg:antichaindecomp} compares every element to at most one smaller element. 
	All other comparisons are between incomparable elements, and every incomparable pair of elements is tested at most once. 
	Thus the time complexity of Algorithm~\ref{alg:antichaindecomp} and hence of the construction of a longest chain is $\O{n + q}$.
\end{proof}

\emph{Computing $s_u$ and~$l_u$.} 
To compute $s_u$ and $l_u$ for all $u \in P \setminus C$ (Line~\ref{line:ac_lisi} in Algorithm~\ref{alg:ideal}), we use the antichain decomposition $A_1, \dots, A_k$ computed before.
Set $s_{c_i} = i$ for $i = 1, \dots, k$.
Let $u \in P \setminus C$ and let $A_i$ be the antichain containing~$u$.
By construction of~$C$, we have $c_i \in A_i$, so $s_u < i$. 
We thus compare $u$ to $c_{i-1}, \ldots, c_1$ until the maximal element on $C$ smaller than $u$ is found.
The values of $l_u$ are determined analogously.
As in Lemma~\ref{lem:antichain_correctness}, the time complexity of this procedure is $\O{n+q}$.
\medskip

\emph{Sorting the lists.}
By relabeling the elements of $P$ in the preprocessing step, we ensure that $u \preceq v$ implies $u \leq v$. 
We need to maintain this property for the recursive calls. 
In the update, we replace $P'$ by $(P' \setminus L_i) \cup S_i$.
To remove~$L_i$ from $P'$, we traverse $P'$ and remove all vertices $u$ with $s_u = i$.
We can construct $S_i$ as a sorted list,  
so adding $S_i$ to $P'$ then amounts to merging two sorted lists.
The time complexity of creating the sorted lists $P_1, \ldots, P_{k}$ is $\O{\sum n_i}$, where $n_i = |P_i|$.
Lemma~\ref{lem:ideal-equiv} implies that $\sum n_i \leq q + n - k \in \O{n+q}$.

	\paragraph{Complexity analysis} 
	
	With the setup described in the previous paragraph, we obtain the following time complexity for the initialization and a single iteration:

	\begin{lemma} \label{lem:ideal_time}
	The initialization of Algorithm \ref{alg:ideal} requires $\O{n^2}$ time. 
		One iteration can be completed in $\O{n+q}$.
		The maximal recursion depth is $n$. In particular, the time complexity of all iterations on a path from the root to a leaf in the recursion tree is $\O{n(n+q)}$.
	\end{lemma}
	\begin{proof}
		In the initialization, Algorithm \ref{alg:ideal} relabels the elements of $P$ in a way that $u \preceq v$ implies $u \leq v$ for all $u,v \in P$. 
		This can be achieved in time $\O{n^2}$ using the classical topological sorting algorithm.
	We showed how to compute the longest chain and construct the sets $L_i$ and $S_i$ (Lines~\ref{line:longestchain} and~\ref{line:lisi}) in time $\O{n+q}$.
		Updating the arrays containing~$P'$ and~$I$ (Lines~\ref{line:updatep} and~\ref{line:updatei}) is done in total time $\O{n+q}$ since every element is added and removed at most once.
		Therefore, the runtime of one iteration of the recursive function is $\O{n+q}$.
		Since~$P_i$ contains no element of $C$, the recursion depth is at most $n$.
	\end{proof}

Concerning the space requirement of Algorithm~\ref{alg:ideal}, we note the following:
	
	\begin{lemma} \label{lem:ideal_space}
		The space complexity of Algorithm \ref{alg:ideal} is $\O{n^2}$.
	\end{lemma}
	\begin{proof}
		Every iteration of Algorithm \ref{alg:ideal} requires $\O{n}$ memory for a longest chain, the subposet~$P'$, the values $s_u$ and $l_u$ and sets $S_i$ and $L_i$.
		These variables have to be stored on the call stack whose height is at most $n$ by Lemma \ref{lem:ideal_time}.
	\end{proof}

We now move to the potential analysis. 
Before we define the relevant notions, we prove the following number-theoretic lemma. 
	
	\begin{lemma}\label{lem:aux}
		For any $a_1, \ldots, a_g, b_1, \ldots, b_f \in \mathbb{Z}_{>0}$, it holds that 
		$$48 \sum_{j=1}^g j \left( \binom{a_j}{3} + 1\right) + 48 \sum_{j=1}^f j\left(\binom{b_j}{3} + 1 \right) \ge \sum_{j=1}^g a_j \cdot \sum_{j=1}^f b_j.$$
	\end{lemma}
	\begin{proof}
		From Jensen's inequality, it follows that
		$$\frac{\sum_{j=1}^g \frac{1}{\sqrt{j}} \left(\sqrt{j}{a_j} \right)^3}{\sum_{j=1}^g \frac{1}{\sqrt{j}}}
		\ge
		\left(\frac{\sum_{j=1}^g \frac{1}{\sqrt{j}} \sqrt{j}{a_j} }{\sum_{j=1}^g \frac{1}{\sqrt{j}}}\right)^3,
		$$
		which implies that 
		$$\sum_{j=1}^g j a_j^3 \ge \frac{\sum_{j=1}^g a_j}{\left(\sum_{j=1}^g \frac{1}{\sqrt{j}}\right)^2} \left(\sum_{j=1}^g a_j\right)^2 \ge \frac14 \left(\sum_{j=1}^g a_j\right)^2$$
		where the last inequality follows from $\sum_{j=1}^g a_j \ge g$ and $\sum_{j=1}^g \frac{1}{\sqrt{j}}  2 \sqrt g$.
		From $2(a_j-1)(a_j-2)+12 \ge a_j^2$, it follows that
		$$12\sum_{j=1}^g j \left( \binom{a_j}{3} + 1 \right) \ge \sum_{j=1}^g j a_j^3.$$
		
		Without loss of generality, we assume that $\sum_{j=1}^g a_j \ge \sum_{j=1}^f b_j$.
		By combining the previous inequalities, we obtain
		$$48 \sum_{j=1}^g j \left( \binom{a_j}{3} + 1\right) + 48 \sum_{j=1}^f j\left(\binom{b_j}{3} + 1 \right) \ge 4\sum_{j=1}^g j a_j^3 \ge \left(\sum_{j=1}^g a_j\right)^2 \ge \sum_{j=1}^g a_j \cdot \sum_{j=1}^f b_j.\qedhere$$
	\end{proof}
	 
	 Our aim is to show that Algorithm~\ref{alg:ideal} has amortized constant delay, using the Pyramid condition from the previous section (see Definition~\ref{def:pyramidcondition}). 
	 To this end, we define a partition $Q = Q' \dcup Q'' \dcup Q'''$ as follows: 
	 Let $Q'$ be the set of incomparable pairs with one vertex in $C$, $Q''$ the set of incomparable pairs $\{u,v\}$ with $u,v \in P \setminus C$ such that $u$ and $v$ have a common incomparable element on $C$, and $Q'''$ the remaining pairs in $Q$. 
	 Recall that $n_i =|P_i|$. 
	 Let $Q_i$ be the set of incomparable pairs in $P_i$ and set $q_i = |Q_i|$.
	 Set $Q_i'' = Q'' \cap Q_i$ and $q_i'' = |Q_i''|$.
	 Analogously, we define $Q'''_i$ and $q'''_i$.
	 Let~$T$ be the set of all 3-antichains in $P$.
	 As above, we define a partition $T = T' \dcup T'' \dcup T'''$ and let $t', t'', t'''$ denote the size of $T', T'', T'''$, respectively.
	 For $i =1, \dots, k$, let $T_i \subseteq T$ denote the 3-antichains contained in $P_i$, and set $t_i = |T_i|$.
	Set $T_i'' = T'' \cap T_i$ and $T_i''' = T''' \cap T_i$, and define $t_i'' = |T_i''|$ and $t_i''' = |T_i'''|$. 
	 
	For the potential analysis, we need the following upper bound on~$q'''$. 
	
		\begin{lemma}\label{lem:pair_ideal}
		It holds that $q''' \leq 96 \left(\sum t''_i + \sum n_i\right)$.
	\end{lemma}
	\begin{proof}
		Let $\{u,v\} \in Q'''$.
		As $u$ and $v$ are incomparable, we have $s_u < l_v$ and $s_v < l_u$, so $\max\set{s_u,s_v} < \min\set{l_u,l_v}$.
		If $\max\set{s_u,s_v} < i < \min\set{l_u,l_v}$ for some $i \in \{1, \dots, k\}$, then $c_i$ is incomparable with both $u$ and~$v$. 
		This contradicts the assumption $\{u,v\} \in Q'''$.
		Therefore, we have $\max\set{s_u,s_v} = \min\set{l_u,l_v} - 1 =: i$.
		Since $l_u \ge s_u + 2$, it follows that $s_u = l_v - 1$ or $s_v = l_u - 1$.
		This yields a partition $Q''' = G_1 \dcup \cdots \dcup G_{k}$, where $G_i$ consists of all pairs ${u,v} \in Q'''$ with $u \in S_i$ and $v \in L_{i+1}$.
		Clearly, we have $|G_i| \le |S_i| |L_{i+1}|$.
		
		We now decompose the subposet of $P$ induced by $S_i$ into antichains.
		Let $A_{i,1}$ be the set of minimal elements of $S_i$. 
		Let $A_{i,2}$ be all minimal elements of $S_i \setminus A_1$. 
		We continue this way until we reach the last non-empty antichain, denoted by $A_{i,{e_i}}$.
		Similarly, let $B_{i,1}$ denote the set of maximal elements of $L_{i+1}$, let $B_{i,2}$ the set of maximal elements in $L_{i+1} \setminus B_{i,1}$, and continue up to the last non-empty antichain $B_{i,f_i}$.
		Let $a_{i,j} = |A_{i,j}|$ and $b_{i,j} = |B_{i,j}|$. 
		Using Lemma \ref{lem:aux}, we derive
		$$|G_i| \le |S_i| |L_{i+1}| \le \sum_{j=1}^{e_i} a_{i,j} \cdot \sum_{j=1}^{f_i} b_{i,j} \le 48 \sum_{j=1}^{e_i} j \left( \binom{a_{i,j}}{3} + 1\right) + 48 \sum_{j=1}^{f_i} j\left(\binom{b_{i,j}}{3} + 1 \right).$$
		
		Let $Z$ be the set of 4-antichains which contain an element of $C$, and let $Z_i$ be the set of 4-antichains in $Z$ containing three elements of $S_i$.
		Since $C$ is a maximal chain of $P$, all elements in $A_{i,j}$ are incomparable with $c_{i+1}, \ldots, c_{i+j}$.
		Thus $Z_i$ contains $j \binom{a_{i,j}}{3}$ 4-antichains that contain three elements from $A_{i,j}$. 
		This implies that 
		$$\sum_{i=1}^{k-1} \sum_{j=1}^{e_i} j \binom{a_{i,j}}{3} \le \sum_{i=1}^{k-1} |Z_i| \le |Z| = \sum t''_i.$$
		
		Since every antichain $A_{i,j}$ is non-empty and $A_{i,j} \subseteq P_{i+1} \cap \dots \cap P_{i+j}$, we obtain \[\sum_{i=1}^{k-1} \sum_{j=1}^{e_i} j \le \sum n_i.\]
		Using a similar bound for $B_{i,j}$, we obtain the desired inequality.
	\end{proof}

	
	We set $\alpha = 922$, $\beta = 921$, $\gamma = 385$, $\delta = 192$ and define a potential function $\Phi(P) = \alpha + \beta n + \gamma q + \delta t$. 
	Using this potential and the Pyramid condition (see Definition~\ref{def:pyramidcondition}), we now show that Algorithm~\ref{alg:ideal} has constant amortized delay.
	
	\begin{lemma} \label{lem:phi-ideal}
		For every non-empty poset $P$, we have $\sum \Phi(P_i) - \Phi(P) \ge n + q$.
		In particular, Algorithm~\ref{alg:ideal} has constant amortized delay.
	\end{lemma}
	
\begin{proof}
We show that $\sum \Phi(P_i) - \Phi(P) - (n+q) \geq 0$. 
First, we analyze the terms in this expression that correspond to poset elements. 
	We have $n \le k + \frac12 \sum n_i$
	as Lemma~\ref{lem:ideal-equiv} yields $2(n-k) \le \sum n_i$.
	Hence,
	$$\beta \sum n_i - \beta (n+1) \ge \frac{\beta-1}2 \sum n_i - (\beta+1) k.$$
	
	Second, we analyze the contribution of incomparable pairs.
For every incomparable pair of the form $\{c_i, v\}$ with $v \in P \setminus C$, we have $v \in P_i$. 
Lemma~\ref{lem:once} yields $q' + n - k = \sum n_i$, so in particular $q' \leq \sum n_i$.
 We have $q'' \le \frac12 \sum q''_i$ as by 
		Lemma \ref{lem:once}, both elements of any pair in $Q''$ lie in at least two of the posets $P_0, \ldots, P_{k}$.
	Finally, Lemma~\ref{lem:once} yields $q''' = \sum q'''_i$.
		Moreover, we have $q''' \le 96 (\sum t''_i + \sum n_i)$
		by Lemma~\ref{lem:pair_ideal}. 
Combining these inequalities yields
	$$\gamma \sum q_i - (\gamma+1) q \ge \frac{\gamma-1}2 \sum q''_i - (\gamma+97) \sum n_i - 96 \sum t''_i.$$
	
	Lastly, we analyze the terms corresponding to 3-antichains.
We have $t' = \sum q''_i$ as for every 3-antichain of the form $\{c_i, u,v\}$ with $u,v \in P \setminus C$, we have $\{u,v\} \in Q_i''$.
We have $t'' \le \frac12 \sum t''_i$ as 
		by Lemma \ref{lem:once}, every 3-antichain $T''$ belong to at least two posets of $P_0, \ldots, P_{k}$.
		Finally, we have $t''' = \sum t'''_i$, which  again follows from Lemma \ref{lem:once}.
	Together, this yields
	$$\delta \left(\sum t_i - t\right) \ge \frac\delta2 \sum t''_i - \delta \sum q''_i.$$
	
Combining the above inequalities, we obtain 
	\begin{align*}
		\sum& \Phi(P_i) - \Phi(P) - n - q \\
		\ge&\; \alpha k + \frac{\beta-1}2 \sum n_i - (\beta+1) k \\
		&+ \frac{\gamma-1}2 \sum q''_i - (\gamma+97) \sum n_i - 96 \sum t''_i \\
		&+ \frac\delta2 \sum t''_i - \delta \sum q''_i \\
		\ge&\; k (\alpha - \beta - 1)
		+ \sum n_i \left(\frac{\beta-1}{2} - \gamma - 97\right)
		+ \sum q''_i \left(\frac{\gamma-1}2 - \delta\right)
		+ \sum t''_i \left(\frac\delta2 - 96\right)\\
		=&\; 0.
	\end{align*}
By Lemma~\ref{lem:ideal_time}, one iteration of Algorithm~\ref{alg:ideal} can be executed in time $\O{n+q}$.
By the above, the potential $\Phi$ satisfies the Pyramid condition for $T^\star =\O{1}$. 
By Theorem~\ref{thm:amortized}, Algorithm~\ref{alg:ideal} has constant amortized delay.
\end{proof}

	The following theorem summarizes the results of this section:
	
	\begin{theorem}
		Algorithm \ref{alg:ideal} enumerates all ideals with constant amortized delay while requiring $\O{n^2}$ space.
		For $k \in \mathbb{N}$, the $k$-th ideal is visited in time $\O{k + n(n+q)}$.
	\end{theorem}
	
	\begin{remark}
	Instead of requiring the input poset to be given by its adjacency matrix, one can modify Algorithm~\ref{alg:ideal} so that the input poset is given by its Hasse diagram. 
	The idea is as follows:
	Let $m$ be the number of edges in the Hasse diagram. 
	We show that a single iteration needs $\O{n+m}$ instructions.
	For the potential analysis, we use a potential of the form $\Phi(P) = \alpha + \beta n + \gamma m + \delta q$. 
	With this, we can show as before that the algorithm achieves constant amortized delay.
	We omit the details here.
	\end{remark}
	
	We conclude this section by showing that we cannot use the Push-out method (see Theorem~\ref{thm:uno}) to show that Algorithm~\ref{alg:ideal} has constant amortized delay.
	
	\begin{example}\label{ex:uno}
	Let $\ell \in \mathbb{N}$. 
	We construct a poset $P$ for which the Push-out method (see Theorem~\ref{thm:uno}) fails to prove that Algorithm~\ref{alg:ideal} has constant amortized delay. 
	See Figure~\ref{fig:unoexample} for an illustration.
	
	\begin{figure}
	\centering
	\includegraphics{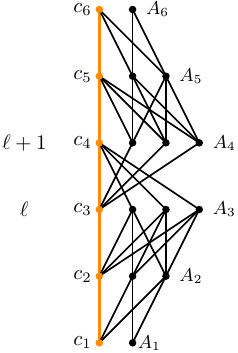}
	\label{fig:unoexample}
	\caption{Construction of the poset $P$ for which the Algorithm~\ref{alg:ideal} cannot be analyzed using the Push-out method. The illustration uses $\ell =3$.}
	\end{figure}
	We start with a disjoint union of $2\ell$ antichains $A_1, \dots, A_{2\ell}$, where the size of $A_{\ell+1-i}$ and $A_{\ell+i}$ is $\lceil 1+\ell/i \rceil$ for $i \le \ell$
	(the rounding does not change the asymptotics in the following computations).
	We add cover relations as follows:
	Fix $c_i \in A_i$ for $i =1, \dots, 2\ell$ and let
   $c_i \prec c_{i+1}$ for $i \in \{1, \dots, 2\ell-1\}$. 
	Thus $\{c_1, \dots, c_{2\ell}\}$ is a longest chain in $P$. 
	For $i\in \{2, \dots, \ell\}$, $a \in A_{i} \setminus \{c_i\}$ and $b \in A_{i-1}$, let $b \prec a$. 
	For all $a \in A_\ell$, let $a \prec c_{\ell+1}$. 
	Symmetrically, for all $i \in \{\ell+1, 2\ell -1\}$,  $a \in A_i \setminus \{c_i\}$ and $b \in A_{i+1}$, let $a \prec b$, and let $c_{\ell} \prec a$ for all $a \in A_{\ell+1}$. 
%
	
	We have $n =2 \sum_{i=1}^\ell \lceil 1 + \ell/i\rceil = \Theta(\ell \log \ell)$.
	The incomparable pairs in $P$ are of the following forms: pairs of elements from the same antichain, all elements in $(A_i \setminus \set{c_i}) \times (A_j \setminus \set{c_j})$ for $i \le \ell < j$, and pairs of the form $\{c_i, a\}$ with $a \in A_j \setminus \{c_j\}$ and $j <i$ if $i \leq \ell-1$, and $j > i$ if $i \geq \ell+1$.
	Thus $q = \Theta(n^2) = \Theta(\ell^2 \log^2 \ell)$.
	For $i \le \ell-1$, we have $P_i =\bigcup_{j=1}^{i+1} A_i \setminus \set{c_i}$.
	Thus $n_i = \Theta(\ell \log i)$ and $q_i = \Theta(\sum_{j=\ell-i}^\ell (\ell/j)^2)$.
	Similar formulas hold for $i \ge \ell+1$.
	Finally, we have $P_\ell =P \setminus C$, so $n_\ell \le n$ and $q_\ell \le q$.
	
	Recall that the complexity of the iteration of Algorithm~\ref{alg:ideal} processing $P$ is $T(P) = \Theta(n+q) = \Theta(\ell^2 \log^2 \ell).$
	The time complexity of processing all children iterations of $P$, except for~$P_\ell$, is \[\sum_{i = 0}^{2\ell} T(P_i) - T(P_\ell) = \Theta\left( \sum_{i=1}^{\ell-1} \sum_{j=\ell-i}^\ell (\ell/j)^2 \right) = \Theta(\ell^2 \log \ell)\] since \[\sum_{i=1}^{\ell-1} \sum_{j=\ell-i}^\ell \frac{1}{j^2} = \sum_{i=1}^\ell i \frac{1}{i^2} - \frac{1}{\ell^2} = \Theta(\log \ell).\]
	Furthermore, $(|C(P)|+1) T^\star = (2\ell+1) T^\star = \Theta(\ell)$.
	Now assume that Algorithm~\ref{alg:ideal} satisfies the Push-out condition for some $\alpha > 1$ and $\beta \geq 0$. 
	In particular, this implies that
	$$\sum_{i = 0}^{2\ell} T(P_i) \geq \alpha T(P) + \beta (2\ell+1) T^\star.$$
This yields
\begin{align*}
	\Theta(\ell^2 \log \ell) &= \sum_{i = 0}^{2\ell} T(P_i) - T(P_\ell) \\
	&\ge \alpha T(P) - T(P_\ell) - \beta (2\ell+1) T^\star \\
	&\ge (\alpha-1) T(P) -  \beta (2\ell+1) T^\star  \\
	&= \Theta(\ell^2 \log^2 \ell)
\end{align*}	
(note that $\alpha -1 > 0$).
By choosing $\ell \in \mathbb{N}$ sufficiently large, we thus obtain a contradiction. 
Thus Algorithm~\ref{alg:ideal} does not satisfy the Push-out condition with constant time.
\end{example}
	
\subsection{Gray code enumeration of poset ideals}\label{sec:graycodeideals}

In this section, we modify Algorithm~\ref{alg:ideal} to list the ideals of the input poset $P$ in Gray code order, while maintaining constant amortized delay.
This is achieved by choosing a particular ordering of the recursive calls in Algorithm~\ref{alg:ideal} as well as specified start and end ideals for every recursive call.  
		
For $k \in \mathbb{N}$, let $\flip^k(P)$ be the following flip graph: the vertices are the ideals of $P$, and two ideals are connected by an edge if they differ in at most $k$ elements.
Pruesse and Ruskey~\cite{pruesse1993gray} showed that $\flip^2(P)$ has a Hamiltonian path.
Note that in order to obtain loopless algorithm for enumerating ideals, it would be sufficient to find an efficiently computable Hamiltonian path in $\flip^m(P)$ for an arbitrary constant $m$.
		
We first show that the order of the recursive calls in Algorithm~\ref{alg:ideal} can be chosen such that two successively visited ideals differ in at most three elements.
We proceed in two steps: first, we show that there is a Hamiltonian path in $\flip^3(P)$ that lists the ideals in $P_i$ consecutively, and maintains this property in the inductive steps. 
Subsequently, we discuss how to convert this to an algorithm. 
We use the notation from Section~\ref{sec:ideal_prelims}.

		
\begin{lemma}\label{lem:ideal_gray}
For every ideal $I \neq P$, the flip graph $\flip^3(P)$ contains a Hamilton path between~$I$ and $P$.
Similarly, for every ideal $I \neq \emptyset$, the flip graph $\flip^3(P)$ contains a Hamilton path between~$I$ and $\emptyset$.

\end{lemma}
\begin{proof}
We prove both statements in parallel, proceeding by induction on $|P|$.
It is easy to see that the statement holds if~$P$ is a chain. 
In the following, we therefore assume that~$P$ contains incomparable elements.
We first assume $I \neq P$ and construct a Hamiltonian path from $I$ to $P$ in $\flip^3(P)$.
Let~$C$ be a longest chain in $P$ and set $\zeta = \zeta(I)$.
We distinguish two cases:

\paragraph{Case 1: $0 \leq \zeta < k$.}
We distinguish the case that $\zeta$ is even or odd. 

First assume that $\zeta$ is even. 
Recall the correspondence between $\id_i(P)$ and $\id(P_i)$ given by Lemma~\ref{lemma:idealcorrespondence}.
Using this correspondence, we list the ideals of $P$, processing the posets $P_0, \dots, P_{k}$ recursively in the order
\begin{equation}\label{eq:zeta_even}
P_\zeta, P_{\zeta-2}, \ldots, P_0, P_1, P_3, \ldots, P_{\zeta-1}, P_{\zeta+1}, P_{\zeta+2}, \ldots, P_{k}.
\end{equation}

For every $i \in \{0, \dots, k\}$, we define ideals $I^f_i$ and $I^l_i$ of $P$ such that $\varphi_i^P(I^f_i)$ and $\varphi_i^P(I^l_i)$ are the first and the last visited ideals of $\id(P_i)$, respectively.
We distinguish multiple cases. 
For $\zeta > 0$, the cases given in Table~\ref{tab:zeta_even}.
Note that all sets listed are ideals of $P$ by Lemma~\ref{lemma:propertiesdi}.
By $P_i^{\text{min}}$ and $P_i^{\text{max}}$, we denote the set of minimal and maximal elements in $P_i$, respectively. 
\renewcommand{\arraystretch}{2.8}
\begin{table}[ht!]
	\small
	\setlength{\tabcolsep}{3pt}
	\centering
	\begin{tabular}{c|c|c}
		$i$ & $I_i^f$ & $I_i^l$ \\
		\hline\hline
		$\zeta$, if $\zeta >0$& $I$ &  $\begin{cases}
			D_\zeta &\text{if $\neq I$ or $P_\zeta = \emptyset$} \\
			D_\zeta \cup \set{x} &\text{else, with $x \in P_\zeta^{\text{min}}$}
		\end{cases}$ \\
		\hline
	\shortstack{$2 \leq i \leq \zeta-2,$\\ $i \equiv \zeta \pmod{2}$}
		& $D_{i+2} \setminus \set{c_{i+1},c_{i+2}}$& $\begin{cases}
			D_i &\text{if $\neq I^f_i$ or $P_i = \emptyset$} \\
			D_i \cup \set{x} &\text{else, with $x \in P_i^{\text{min}}$}
		\end{cases}$\\
		\hline
		$0$ & \dittotikz &$
		\begin{cases}
			U_0 & \text{if $\neq I^f_0$ or $P_0 = \emptyset$} \\
			U_0 \setminus \set{x} & \text{else, with $x \in P_0^{\text{max}}$} \\
		\end{cases}$ \\
		\hline
		$1$ &$
		\begin{cases}
			U_0 \cup \set{c_1} & \text{if $\neq U_1$ or $P_1 = \emptyset$} \\
			U_0 \cup \set{c_1} \setminus \{x\} & \text{else, with $x \in P_1^{\text{max}}$} \\
		\end{cases}$& $U_1$\\
		\hline
\shortstack{$3 \leq i \leq\zeta-1,$\\ $i \not \equiv \zeta \pmod{2}$}& $\begin{cases}
			U_{i-2} \cup \set{c_i,c_{i-1}} &\text{if $\neq U_i$ or $P_i = \emptyset$} \\
			U_{i-2} \cup \set{c_i,c_{i-1}} \setminus \{x\} &\text{else, with $x \in P_i^{\text{max}}$} \\
		\end{cases}$ & $U_i$\\
		\hline
		$\zeta + 1, \text{ if } \zeta >0$ & \dittotikz & \dittotikz \\
		\hline
		$\zeta +2 \leq i \leq k$ &$\begin{cases}
			U_{i-1} \cup \set{c_i} & \text{if $\neq U_i$  or $P_i = \emptyset$} \\
			U_{i-1} \cup \set{c_i} \setminus \{x\} & \text{else, with $x \in P_i^{\text{max}}$} \\
		\end{cases}$& \dittotikz
	\end{tabular}
	\caption{Case distinction for the start and end ideals in case $\zeta$ is even. The choice of the element $x$ in $P_i^{\text{min}}$, or $P_i^{\text{max}}$ respectively, is arbitrary.}
	\label{tab:zeta_even}
\end{table}
For $\zeta = 0$, the ordering is given by $P_0, P_1, \dots, P_k$.
We use 
$$I_0^f = I \text{ and } I^l_0 =
\begin{cases}
	U_0 & \text{ if $I \neq U_0$ or $P_0 = \emptyset$} \\
	U_0 \setminus \set{x} & \text{ otherwise, where $x \in P_0^\text{max}$,} \\
\end{cases}$$
as well as
\[I_1^f =
\begin{cases}
	U_0 \cup \set{c_1} & \text{if $\neq I^l_1$ or $P_1 = \emptyset$} \\
	U_0 \cup \set{c_1} \setminus \set{x} & \text{else, where $x \in P_1^{\text{max}}$} \\
\end{cases} \text{ and } I_1^l =U_1\]
(so the lines corresponding to $0$ and $1$ in Table~\ref{tab:zeta_even}, with the change that $I_0^f = I$).
For $i \geq 2$, the ideals are defined according to the case $i > \zeta +1$ in Table~\ref{tab:zeta_even}.

Now let $\zeta \in \{0, \dots, k-1\}$. 	
We show that by induction that $\flip^3(P)$ contains the desired Hamiltonian path. 
For every $i \in \{0, \dots, k\}$, one of $I_i^f$ and $I_i^l$ equals $D_i$ or $U_i$, so its image under~$\varphi_i^P$ is $\emptyset$ or $P_i$. 
By construction, we have $I_i^f = I_i^l$ if and only if $P_i = \emptyset$, in which case $P_i$ has a single ideal.
From Table~\ref{tab:zeta_even}, it is easily checked that whenever $P_i$, $P_j$ appear consecutively in the ordering given by~\eqref{eq:zeta_even}, then $I_i^l$ and $I_j^f$ differ in at most three elements.
By induction, we obtain a Hamiltonian path in $\flip^3(P_i)$ that starts with $\varphi_i^P(I_i^f)$ and ending with $\varphi_i^P(I_i^l)$.
Concatenating the respective preimages yields a Hamiltonian path from~$I$ to $P$ in $\flip^3(P)$.

Now assume that $\zeta$ is odd. We process the posets in the order
\begin{equation}\label{eq:zeta_odd}
	P_\zeta, P_{\zeta-2}, \ldots, P_1, P_0, P_2, \ldots, P_{\zeta-1}, P_{\zeta+1}, P_{\zeta+2}, \ldots, P_{k}.
\end{equation}

For $i \in \{0, \dots, k\}$, we define start and end ideals $I_i^f$ and $I_i^l$ as before.
Compared to the previous case, we only modify the definitions for $i \in \{0,1\}$ as the respective posets are traversed in opposite order in \eqref{eq:zeta_even} and \eqref{eq:zeta_odd}. 
We set
\[I_1^f =D_3 \setminus \set{c_1, c_2} \text{ and } I_1^l = \begin{cases}
	U_1 & \text{if $\neq I^f_1$ or $P_1 = \emptyset$} \\
	U_1 \setminus \set{x} & \text{else, with $x \in P_1^{\text{max}}$} \\
\end{cases},\] 
as well as
\[I_0^f = \begin{cases}
	U_1 \setminus \{c_1\} & \text{if $\neq U_0$ or $P_0 = \emptyset$} \\
	U_1 \setminus \{c_1, x\} & \text{otherwise, with $x \in P_0^{\text{max}}$} \\
\end{cases} \ \text{ and } I_0^l = U_0.\]

As in the even case, it can be shown that consecutively listed ideals differ in at most three elements and every recursive call starts or ends with the empty set or the full poset, so we obtain a Hamiltonian path from $I$ to $P$ in $\flip^3(P)$ by induction. 

\paragraph{Case 2: $\zeta = k$.}
In this case, the Hamiltonian path begins and ends with ideals belonging to $P_k$ via the correspondence from Lemma~\ref{lemma:idealcorrespondence}. 
Thus we need to split $\id(P_k)$, traversing one part at the beginning of the enumeration and the other at the end. 
To this end, we fix $y \in P \setminus I$ 
and let $D_y =\{u \in P_k \colon u \preceq y\}$ be the downset of $y$ in $P_k$.
	We define two subposets $P' = \{u \in P_{k} \colon u \not \succeq y\}$ and $P'' = \{u  \in P_{k} \colon u \not \preceq y\}$.
			Set $\hat{\id}(P'') = \{J \cup D_y \colon J \in \id(P'')\}$.
			It easily verified that \[\id(P_k) = \id(P') \dcup \hat{\id}(P'').\]		

			If $k$ is even, we enumerate the ideals of $P$ in the order 
			$$\id(P'),\, \id(P_{k-2}),\, \id(P_{k-4}), \ldots, \id(P_0),\, \id(P_1),\, \id(P_3), \ldots,\, \id(P_{k-1}), \hat{\id}(P'').$$ 
			If $k$ is odd, we use the ordering
			$$\id(P'),\, \id(P_{k-2}),\, \id(P_{k-4}), \ldots, \,\id(P_1), \,\id(P_0), \,\id(P_2), \ldots,\, \id(P_{k-1}),\, \id(P'').$$
			
			We use the same start and end ideals for $P_0, \ldots, P_{k-1}$ as in the case $\zeta < k$.
			The first enumerated ideal for $P'$ is 
			$I'^f = I$, the last is 
			$$I'^l =
			\begin{cases}
				D_{k} & \text{ if $\neq I$ or $P' = \emptyset$} \\
				D_{k} \cup \set{x} & \text{else, where $x \in P'^{\text{min}}$.}
			\end{cases}$$
		The first enumerated ideal for $P''$ is
			$$I''^f =
			\begin{cases}
				P & \text{ if $P'' = \emptyset$} \\
				P \setminus \set{x} & \text{else, with $x \in P''^{\text{max}}$}, \\
			\end{cases}$$
		the last is  $I''^l = P$.
		With this, we proceed as in the previous case.
		\end{proof}
		
		We now describe how to incorporate the choice of the recursion order from the previous proof in Algorithm~\ref{alg:ideal}. 
		For the sake of clarity, we only highlight the implementation of the rules given by the above proof.
		
		We modify the \texttt{Enumerate} routine in Algorithm~\ref{alg:ideal} such that it takes four input arguments:
		the poset $P$, elements $R$ in an input poset, the first visited ideal $I^f$ and the last visited ideal $I^l$ of~$P$. 
		The recursive function visits $R \cup I$ for all ideals $I$ of $P$.
		
		We initialize $D \gets \bigcup_{i=2}^\zeta L_i$ (so $D = D_\zeta$), 
		$R' \gets R \cup D$, $P' \gets \bigcup_{i=2}^\zeta S_i \setminus D$ (so $P' = P_\zeta$), and 
		$J^f \gets I^f \setminus D$ (so $J^f= I^f_\zeta \setminus D_\zeta$). 
		We then enumerate the ideals in $P_\zeta$ by calling $\texttt{Enumerate}(P', R', J^f, J^l)$. 
		
		To enumerate the ideals corresponding to $P_{\zeta-2}$, we update the variables as follows:
		\begin{alignat*}{1}
P' &\gets (P' \setminus (S_{\zeta-1} \cup S_{\zeta-2})) \cup (L_{\zeta-1} \cup L_{\zeta-2}) \textbf{ // } P' = P_{\zeta-2} \\
J^f &\gets (S_{\zeta-1} \cup S_{\zeta-2}) \setminus \set{c_{\zeta-1}, c_{\zeta-2}} \textbf{ // } J^f = I^f_{\zeta-2} \setminus D_{\zeta-2} \\
J^l &\gets
\begin{cases}
\emptyset & \text{ if $J^f \neq \emptyset$ or $P' = \emptyset$} \\
\set{x} & \text{else, where $x \in P'^{\text{min}}$.}
\end{cases} \\
R' &\gets R' \setminus (S_{\zeta-1} \cup S_{\zeta-2}) \textbf{ // }  R' = R \cup D_{\zeta-2} \\
D &\gets D \setminus (S_{\zeta-1} \cup S_{\zeta-2}) \textbf{ // } D = D_{\zeta-2}.
		\end{alignat*}
We then call $\texttt{Enumerate}(P',R',J^f,J^l)$. 
The subsequent changes described by the proof of Lemma~\ref{lem:ideal_gray} can be implemented analogously.

To determine a minimal or maximal element $x$ in $P_i$, we use the computation of the longest chain in the children iteration for $P_i$, and let $x$ be the minimal or maximal element of this chain, respectively. 
As in Lemma~\ref{lem:ideal_time}, we can argue that every element of $P$ is inserted or removed a constant number of times.
Thus the time complexity of a single iteration of this modified version of Algorithm~\ref{alg:ideal} is $\O{n+q}$.  
In the complexity analysis (Lemma~\ref{lem:phi-ideal}), we replace $\Phi(P_k)$ by $\Phi(P') + \Phi(P'')$ due to splitting the poset $P_k$ in the second case. 
		Since $\Phi(P_k) \leq \Phi(P') + \Phi(P'')$, the proof of Lemma~\ref{lem:phi-ideal} can be carried out as before. 
This shows that the algorithm has constant amortized delay. 
By Theorem~\ref{thm:loopless}, we can transform it into an algorithm 
Summarizing, we obtain the following:
	
\begin{theorem}
With the above-described recursive order, Algorithm~\ref{alg:ideal} enumerates all ideals of a poset $P$ with constant delay such that two consecutively visited ideals differ in at most three elements. 
It requires $\O{n(n+q)}$ space and time for initialization.
	\end{theorem}

	\section{Constant amortized time enumeration of antichains}\label{sec:antichains}
	
	In this section, we describe a recursive algorithm (Algorithm~\ref{alg:antichains}) that enumerates all antichains of an input poset $(P, \preceq)$ with constant delay. 
	It is closely related to the enumeration algorithm for ideals described in the previous section, but uses a simpler recursive setup.
	Again, we first present a basic version and then show how to refine this to an algorithm listing the antichains in Gray code order.
	
	\paragraph{Recursive structure}
	We fix a longest chain $C$ in $P$ and write $C = \{c_1, \dots, c_{k}\}$ with $c_1\prec \dots \prec c_{k}$. 
	Note that every antichain of $P$ contains at most one element of $C$.
	For $i = 1, \ldots, k$, we let $P_i$ denote the set of all elements of $P$ incomparable with $c_i$, and let $P_r = P \setminus C$.
	For $i = 1, \dots, k$, there is a natural bijection between the set of antichains of~$P$ containing $c_i$ and antichains in $P_i$. 
	Moreover, the antichains of $P$ that do not contain an element of $C$ are precisely the antichains of $P_r$. 
	Thus this yields a natural bijection between antichains of $P$ and the union of all antichains of $P_1, \ldots, P_{k}, P_r$, which we exploit for the recursion.
	
	\paragraph{Description of the algorithm}
	In the initialization phase, we relabel the elements of~$P$ to guarantee that $u \preceq v$ implies $u \leq v$. 
	We then proceed recursively as follows: to construct an antichain, we maintain an array $A$ of elements that were already added to the antichain under construction by parent iterations. 
	In each iteration, we first find a longest chain $C$ in the currently considered poset~$P$, and construct the auxiliary sets $L_i$ and $S_i$ defined below. 
	As any antichain contains at most one element from $C$, we distinguish which element of~$C$ (if at all) is contained. 
	That is, we first add the element~$c_1$ to $A$ and call the recursion on the subposet of $P$ consisting of the elements incomparable with~$c_1$. 
	We then proceed analogously with the elements $c_2, \dots, c_{k}$. 
	It remains to enumerate the antichains not containing an element of $C$. 
	For this, we call the recursion with the current antichain $A$ on the subposet consisting of the elements in $P \setminus C$. 
	
	\paragraph{Implementation details}
	The input of Algorithm~\ref{alg:antichains} is a poset $P$ on $n$ elements. 
	As in Section~\ref{sec:ideals},~$P$ is given by a matrix $M$ with $M_{ij} =1$ if $i \prec j$, and $M_{ij}=0$ otherwise. 
	In a preprocessing step, we relabel the elements of $P$ such that $i\preceq j$ implies $i \leq j$.
	For the recursion, note that the initial ordering of all elements is preserved in every subposet.
	Thus, the recursive function calls can use the same matrix $M$ for comparison. 
	The visited antichains are stored in an array.
	As in Section~\ref{sec:ideals}, we set
	\[s_u = \max\{i \in \{0, \dots, k\} \colon c_i \prec u\}\text{ and } l_u = \min\{i \in \{0, \dots, k\} \colon u \prec c_i\}.\]	
	For $i \in \{1, \dots, k\}$, we set $s_{c_i} = l_{c_i} = i$.
	For $i \in \set{0, \ldots, k}$, let 
	\[S_i = \set{u \in P \setminus C \colon s_u = i} \text{ and } L_i = \set{u \in P \setminus C \colon l_u = i}\]

	A longest chain (Line~\ref{line:ac_chain}) and the sets $L_i$ and $S_i$ (Line~\ref{line:ac_lisi}) are computed as in Section~\ref{sec:ideals}, using Algorithm~\ref{alg:antichaindecomp}. 

\paragraph{Complexity analysis}
We now analyze the time and space complexity of Algorithm~\ref{alg:antichains}. 
In particular, we show that it achieves constant amortized delay.
We first analyze the time complexity of the initialization and of a single iteration.
Again, let $Q$ denote the set of incomparable pairs of~$P$ and set $q =|Q|$.

	\begin{lemma}\label{lemma:antichain_time}
	The initialization of Algorithm \ref{alg:antichains} requires $\O{n^2}$ time. 
	The time complexity of one iteration is $\O{n+q}$.

	\end{lemma}
	\begin{proof}
	The initialization is the same as for Algorithm~\ref{alg:ideal}.
		The computation of the longest chain (Line~\ref{line:ac_chain}) and the values $s_u$ and $l_u$ (Line~\ref{line:ac_lisi}) can be carried out in $\O{n+q}$ (see Section~\ref{sec:ideals}).
%
		The remaining non-recursive steps can be carried out in time $\O{n+q}$.
		Therefore, one call of the recursive function in Algorithm~\ref{alg:antichains} has time complexity $\O{n + q}$.
	\end{proof}
	
	Similarly as for Algorithm \ref{alg:ideal}, one proves the following:
	
	\begin{lemma}
 Algorithm~\ref{alg:antichains} has maximal recursion depth $n$.
 The time complexity of the iterations on a path from the root to a leaf in the recursion tree is $\O{n(n+q)} \subseteq \O{n^3}$, using $\O{n^2}$ memory.
	\end{lemma}
	
	We now proceed to the proof that Algorithm~\ref{alg:antichains} achieves constant amortized delay. 
	As in Section~\ref{sec:ideals}, we need a finer decomposition of the occurring sets. 
	
	Let $n_i$ and $n_r$ denote the size of $P_i$ and $P_r$, respectively.
	For $i =1, \dots,k$, let $Q_i$ denote the set of incomparable pairs $\{u,v\}$ with $u,v \in P_i$ and set $q_i = |Q_i|$. Analogously, we define $Q_r$ and $q_r$.
	We define a partition $Q = Q' \dcup Q'' \dcup Q'''$ as in Section~\ref{sec:ideals}:
	Let $Q'$ be the set of incomparable pairs in~$P$ containing an element of $C$. 
	Let $Q''$ be the incomparable pairs in $P \setminus C$ which have a common incomparable element on $C$. 
	Finally, let $Q'''$ denote the remaining incomparable pairs, i.e, those unextendable to an antichain of size 3 by any element of $C$.
	Let $q'$, $q''$, $q'''$ denote the sizes of $Q'$, $Q''$, and $Q'''$, respectively. 
	Set $Q''_i = Q'' \cap Q_i$ and $Q'''_i = Q''' \cap Q_i$.
	Note that $Q'' \cup Q''' = Q_r$ and $Q' \cap Q_i = \emptyset$.
	
	Let $T$ be the set of all 3-antichains in $P$.
	As above, we define a partition $T = T' \dcup T'' \dcup T'''$ and let $t', t'', t'''$ be the sizes of $T', T'', T'''$, respectively.
	For $i =1, \dots, k$, let $T_i \subseteq T$ denote those 3-antichains with all elements in $P_i$, and set $t_i = |T_i|$.
	Analogously, we define $T_r$ and~$t_r$.
	Set $T_i'' = T'' \cap T_i$ and $T_i''' = T''' \cap T_i$, and define $t_i'' = |T_i''|$ and $t_i''' = |T_i'''|$. 
	Analogously, we define $T_r'', T_r'''$ and $t_r'', t_r'''$. 

	\begin{lemma}\label{lem:pair}
		It holds that $q''' \leq 96 \left(\sum t''_i + \sum n_i\right)$.
	\end{lemma}
	
	\begin{proof}
	Although the definition of the recursion sets $P_0, \dots, P_{k}, P_r$ is different than in Section~\ref{sec:ideals}, the proof of Lemma~\ref{lem:pair_ideal} verbally carries over to this setting.
	\end{proof}

 Set $\alpha = 196$, $\beta = 195$, $\gamma = 97$, $\delta = 96$, $\mu = 392$, and define a potential function $\Phi(P) = \alpha + \beta n + \gamma q + \delta t$.
Using this for the Pyramid condition (see Definition~\ref{def:pyramidcondition}), we now show that Algorithm~\ref{alg:antichains} has constant amortized delay.

\begin{lemma} \label{lem:phi-antichain}
For every poset $P$ with $n \geq 2$, we have $\sum \Phi(P_i) + \Phi(P_r) - \Phi(P) \ge n + q$.
If $n = 1$, we have $\mu - \Phi(P) \geq 1$.
In particular, Algorithm~\ref{alg:antichains} has constant amortized delay.
	\end{lemma}

	\begin{proof}
		We show that $\sum \Phi(P_i) + \Phi(P_r) - \Phi(P) - n - q \geq 0$ if $n \geq 2$.
		We first analyze the contribution of terms corresponding to poset elements in this expression.
		We have $n = n_r + k$, as every element of $P$ is either contained $C$ or in $P_r = P \setminus C$. 
		Moreover, we have $n \le k + \sum n_i$, as every element of $P \setminus C$ belongs to at least one subposet $P_i$.
		Together, this yields
		$$\beta (n_r + \sum n_i) - (\beta + 1) n \ge (\beta - 1) \sum n_i - (\beta + 1) k.$$
		
		Next, we analyze the contribution of terms corresponding to incomparable pairs. We have $q' = \sum n_i$, as for every incomparable pair $\{c, v\}$ with $c \in C$ and $v \in P \setminus C$, the element $v$ belongs to $P_i$. 
		Clearly, we have $q'' = q''_r$. 
		Moreover, as every incomparable pair of $Q_r$ is an incomparable pair in $P \setminus C$ and vice versa, we obtain $q''' = q'''_r$. 
		Finally, we have $q'' \le \sum q''_i$, as every incomparable pair $u,v \in P \setminus C$ is incomparable to some $c_i$ belongs to $Q''_i$.
		Together with Lemma \ref{lem:pair}, this yields
		$$\gamma (q_r + \sum q_i) - (\gamma + 1) q \ge (\gamma-1)\sum q''_i - (\gamma + 97) \sum n_i - 96 \sum t''_i.$$
		
		Now we analyze the contribution of the terms corresponding to 3-antichains.
		We have $t' = \sum q''_i$, as by definition, $t'_i = q''_i$.
		Clearly, we have $t'' = t''_r$. 
		Moreover, we have $t''' = t'''_r$, as every 3-antichain of $Q_r$ is an antichain in $P \setminus C$ and vice versa.
		Combining these equalities, we obtain
		$$\delta (t_r + \sum t_i - t) \ge \delta (\sum t''_i - \sum q''_i).$$
		
		Combining the above inequalities yields
		\begin{align*}
			\Phi(P_r&) + \sum \Phi(P_i) - \Phi(P) - n - q \\
			\ge&\; \alpha k + (\beta - 1) \sum n_i - (\beta + 1) k + (\gamma-1)\sum q''_i - (\gamma + 97) \sum n_i \\
			&- 96 \sum t''_i + \delta (\sum t''_i - \sum q''_i) \\
			\ge&\; k (\alpha - \beta - 1) + \sum n_i (\beta - \gamma - 98) + \sum t''_i (\delta - 96) + \sum q''_i (\gamma - \delta - 1) \\
			=&\; 0,
		\end{align*}
	and thus 
	$\Phi(P_r) + \sum \Phi(P_i) - \Phi(P) \geq n+q$.
If $n =1$, we have 
		By Lemma~\ref{lemma:antichain_time}, a single iteration of Algorithm~\ref{alg:antichains} can be executed in time $\O{n+q}$. 
		Let $\mu \geq \alpha + \beta +1$.
		Then if $n=1$, we have $\mu - \Phi(P) \geq 1$.
		By this and the above, the potential $\Phi$ satisfies the Pyramid condition for $T^\star =\O{1}$. 
		By Theorem~\ref{thm:amortized}, Algorithm~\ref{alg:antichains} has constant amortized delay.
	\end{proof}
	
	\begin{remark}
Similarly to Example~\ref{ex:uno}, one can show that Uno's Push-out condition cannot be used to show that Algorithm~\ref{alg:antichains} achieves constant amortized delay.
	\end{remark}

	Similarly as for ideals, we now refine Algorithm~\ref{alg:antichains} to list the antichains in Gray code order. 
	For $k \in \mathbb{N}$, let $\flip^k(P)$ be the flip graph on the set of antichains of $P$, where two antichains are connected by an edge if they differ in at most $k$ elements.
	It is easy to see that $\flip^1(P)$ does not have a Hamiltonian path if $P$ is a chain on at least three elements. 
	Unlike in the case of ideals, it is not known whether $\flip^2(P)$ contains a Hamiltonian path.
	We now present a simple proof that $\flip^3(P)$ has a Hamiltonian path, which we subsequently turn into an algorithm.
	
	\begin{lemma}\label{lem:graycodeantichain}
		For every non-empty poset $P$, the graph $\flip^3(P)$ contains a Hamiltonian path from a 1-antichain to the empty antichain.
	\end{lemma}

	\begin{proof}
		We prove the statement by the induction on the size of $P$.
		If $|P|=1$, the statement clearly holds, so assume that $|P| \ge 2$.
		Let $C$ be a longest chain in $P$.
		For every $i = 1, \ldots, |C|$, we obtain a Hamiltonian path $(A_1^i, \dots, A_{t_i}^i)$ in $\flip^3(P_i)$ that starts with the empty antichain and ends with a 1-antichain $\{a_i\}$ for some $a_i \in P_i$.
		Then $(\{c_i\} \cup A_1^i, \dots, \{c_i\} \cup A_{t_i}^i) = (\{c_i\}, \dots, \{c_i,a_i\})$ is a listing of all antichains of $P$ containing $c_i$. 
		Analogously, we obtain a Hamiltonian path $(A_1^r, \dots, A_{t_r}^r)$ in $\flip^3(P_r)$ that starts with a 1-antichain $\{a_r\}$ and ends with the empty antichain.
		Concatenating these listings, we obtain a Hamiltonian path 
		\[(\{c_1\},\dots, \{c_1, a_1\}, \{c_2\}, \dots, \{c_{|C|}, a_{|C|}\}, \{a_r\}, \dots, \emptyset)\]
		from the 1-antichain $\{c_1\}$ to the empty antichain in $\flip^3(P)$.
		Note that if $P_i$ is empty, only the antichain $\{c_i\}$ is visited.
		Similarly, if $P \setminus C$ is empty, only the empty antichain is visited in the last step.
	\end{proof}

	As the induction in the proof of Lemma~\ref{lem:graycodeantichain} uses the recursive structure of Algorithm \ref{alg:antichains}, 
	we can refine Algorithm~\ref{alg:antichains} to visit the antichains in this ordering, by choosing a particular ordering of the recursive calls.
	The implementation as well as the argument that the time complexity is maintained is similar to the corresponding reasoning for ideals in Section~\ref{sec:graycodeideals}. 
	Summarizing, we obtain the following result:

	\begin{theorem}
		With the above-described recursive order, Algorithm~\ref{alg:antichains} enumerates all ideals of a poset $P$ with constant delay such that two consecutively visited ideals differ in at most three elements. 
		It requires $\O{n(n+q)}$ space and time for initialization.
		The first $k$ antichains are visited in time $\O{k + n(n+q)}$.
	\end{theorem}
		
	\section*{Acknowledgements}
	Sofia Brenner acknowledges funding by a postdoc fellowship by the German Academic Exchange Service (DAAD).
	
	
	\bibliographystyle{amsalpha}
	\bibliography{../bibliography}

\end{document}